\documentclass[11pt,letterpaper]{article}

\usepackage[utf8]{inputenc}
\usepackage[T1]{fontenc}
\usepackage[margin=1in]{geometry}
\usepackage{amsmath,amssymb,amsthm}
\usepackage{booktabs,array,tabularx}
\newcolumntype{Y}{>{\raggedright\arraybackslash}X}
\newcolumntype{P}[1]{>{\raggedright\arraybackslash}p{#1}}
\usepackage{enumitem}
\usepackage{microtype}
\usepackage[numbers,sort&compress]{natbib}
\usepackage[hidelinks]{hyperref}

\allowdisplaybreaks
\setlist{nosep}

\theoremstyle{plain}
\newtheorem{theorem}{Theorem}[section]
\newtheorem{proposition}[theorem]{Proposition}
\newtheorem{lemma}[theorem]{Lemma}
\newtheorem{corollary}[theorem]{Corollary}
\theoremstyle{definition}
\newtheorem{definition}[theorem]{Definition}
\newtheorem{remark}[theorem]{Remark}
\newtheorem{example}[theorem]{Example}
\newtheorem{openproblem}[theorem]{Open Problem}

\newcommand{\F}{\mathbb{F}_2}
\newcommand{\E}{\mathbb{E}}
\newcommand{\Prob}{\mathbb{P}}
\newcommand{\cC}{\mathcal{C}}
\newcommand{\cE}{\mathcal{E}}
\newcommand{\cR}{\mathcal{R}}
\newcommand{\cW}{\mathcal{W}}
\newcommand{\cS}{\mathcal{S}}
\newcommand{\supp}{\operatorname{supp}}
\newcommand{\dist}{\operatorname{dist}}
\newcommand{\rank}{\operatorname{rank}}
\newcommand{\im}{\operatorname{im}}
\newcommand{\Shat}{\operatorname{Shat}}
\newcommand{\1}{\mathbf{1}}
\newcommand{\abs}[1]{\left|#1\right|}
\newcommand{\norm}[1]{\left\lVert#1\right\rVert}
\newcommand{\set}[1]{\left\{#1\right\}}
\newcommand{\ceil}[1]{\left\lceil#1\right\rceil}
\newcommand{\floor}[1]{\left\lfloor#1\right\rfloor}
\newcommand{\qbinom}[2]{\genfrac{[}{]}{0pt}{}{#1}{#2}_{2}}

\hypersetup{
  pdftitle={Cofilling Shattering: A Syndrome-Support Hierarchy for Check Erasures},
  pdfauthor={Joshua Steier},
  pdfsubject={Map-coupled syndrome support, generalized covering radii, check erasures, and top-face erasures},
  pdfkeywords={generalized covering radii, generalized Hamming weights, parity-check matrices, syndrome spaces, support matroids, simplicial cohomology}
}

\title{Cofilling Shattering: A Syndrome-Support Hierarchy for Check Erasures}
\author{Joshua Steier\\[-1mm]
\small Independent Researcher\\[-1mm]}
\date{\today}

\begin{document}
\pagestyle{myheadings}
\markboth{J. Steier}{Cofilling Shattering}
\maketitle

\begin{abstract}
Let $A:\F^n\to\F^m$ be a binary linear map with fixed coordinate bases,
let $C_A=\ker A$, and let $\lambda_A(y)$ be the minimum Hamming weight of a
preimage of the syndrome $y$.  We define $\Shat_{q,s}(A)$ as the least common
check support of a $q$-dimensional syndrome subspace whose every nonzero
element has coset-leader weight at least $s$.  It therefore distinguishes
release of $q$ independent syndromes from release of a subspace with no
easy linear combination.  Deleting check coordinates $F$ releases
$\ker A_{\bar F}/\ker A$, canonically isomorphic to $(\im A)[F]$.

Finiteness implies
\[
 R_q(C_A)\ge \mathsf N_2(q,s),
\]
where $\mathsf N_2(q,s)$ is the shortest length of a binary code of dimension
$q$ and distance at least $s$; profile--Griesmer bounds independently control
common check support.  The hierarchy is coordinate-relabeling invariant but
can change under a change of check basis.  For the pair-repetition code
$C_n=\{(x,x):x\in\F^n\}$, the standard realization $H_0=[I_n\ I_n]$ has
$\Shat_{q,s}(H_0)=\mathsf N_2(q,s)$ whenever feasible.  For every $q\ge1$
and $s\ge2$, with $n=\mathsf N_2(q,s)$, a row-equivalent realization of the
same code has value $q$.

For a simplicial coboundary map $A=\delta_k$, check erasure is top-face
erasure and the released quotient is emergent cohomology.  At $s=1$ the
hierarchy reduces to generalized Hamming weights and is Tutte-determined;
for $s\ge2$, even identical labeled cut codes can have different values.
\end{abstract}

\medskip
\noindent\textbf{Keywords:} generalized covering radii; generalized Hamming weights;
parity-check matrices; syndrome spaces; support matroids; simplicial cohomology.

\smallskip
\noindent\textbf{Mathematics Subject Classification (2020):}
Primary 94B05; Secondary 05B35, 05E45, 05C40, 55U10.

\section{Introduction}\label{sec:introduction}

A parity-check matrix contains more information than its kernel code.  Its
rows specify the check generators that are physically or combinatorially
available, and an erasure of one row is not the same operation after an
arbitrary change of row basis.  This distinction is familiar from
matrix-dependent stopping-set phenomena \cite{schwartzvardy2006}.  Here it
leads to a different question: how many check coordinates must be erased or
activated before one releases a syndrome subspace whose every nonzero
combination is difficult to realize by a low-weight error?

Let $A:\F^n\to\F^m$ be a binary linear map with fixed domain and codomain
coordinates.  Write $C_A=\ker A$ and $\cS_A=\im A$.  For $y\in\cS_A$, let
\[
 \lambda_A(y)=\min\{\abs{x}:Ax=y\}.
\]
This is the coset-leader weight of $y$ for $C_A$.  We study
\[
 \Shat_{q,s}(A)=
 \min\left\{\abs{\supp U}:
 \begin{array}{l}
 U\le \cS_A,\ \dim U=q,\\
 \lambda_A(y)\ge s\text{ for every }0\ne y\in U
 \end{array}\right\}.
\]
The objective is measured in check coordinates, whereas localization is
measured in the quotient metric on variable coordinates.  Generalized
Hamming weights control the first coordinate after the syndrome space has
been chosen; generalized covering radii control how requested syndromes are
generated in the second.  Cofilling shattering couples the two.

The distinction is operational.  A released $q$-space can have a basis of
high-weight representatives and still contain a low-weight linear
combination.  Requiring every nonzero syndrome to have coset-leader weight at
least $s$ rules out that cancellation.  Conversely, row-equivalent
parity-check matrices can define the same code and have the same generalized
covering radii while exhibiting different vulnerability to erasure of the
listed check generators.

The motivating structured realization is the last coboundary map of a
simplicial complex.  If $X$ has dimension $k+1$ and
$F\subseteq X(k+1)$ is deleted, the surviving map is obtained from
$\delta_k$ by deleting the rows indexed by $F$.  We prove
\[
 H^k(X-F;\F)/H^k(X;\F)\cong \im\delta_k[F],
\]
so top-face erasure releases a shortened syndrome subspace.  The rank-only
threshold is therefore a generalized Hamming weight.  The threshold $s$
asks for more: every nonzero emergent class must remain far from the old
cocycle space.

\paragraph{Main results.}
\begin{enumerate}[label=\textbf{R\arabic*.},leftmargin=*,itemsep=3pt]
\item \textbf{An exact check-erasure dictionary.}
Deleting check coordinates $F$ gives a canonical isomorphism
$\ker A_{\bar F}/\ker A\cong(\im A)[F]$.  The hierarchy is defined for
arbitrary binary maps and is invariant under independent permutations of
variable and check coordinates.  Arbitrary row operations are not declared
equivalent because they replace the listed check generators by linear
combinations.  At $s=1$, $\Shat_{q,1}(A)=d_q(\im A)$.

\item \textbf{Covering-radius and support obstructions.}
If $C_A=\ker A$ and $\Shat_{q,s}(A)<\infty$, then
\[
 R_q(C_A)\ge \mathsf N_2(q,s)\ge \mathsf G_q(s),
\]
and $R_1(C_A)\ge\ceil{\mathsf N_2(q,s)/q}$.  Independently,
\[
 \Shat_{q,s}(A)\ge
 \max\left\{d_q(\im A),\mathsf G_q(\Sigma_A(s))\right\}.
\]

\item \textbf{Sharp dependence on the check basis.}
Let $C_n=\{(x,x):x\in\F^n\}$ and $H_0=[I_n\ I_n]$.  Then
$\Shat_{q,s}(H_0)=\mathsf N_2(q,s)$ whenever feasible.  For every
$q\ge1$ and $s\ge2$, set $n=\mathsf N_2(q,s)$.  There is an invertible
$T$ such that $H_1=TH_0$ defines the same code but
\[
 \Shat_{q,s}(H_1)=q<n=\Shat_{q,s}(H_0).
\]
For $q\ge2$ this is a genuinely collective separation.

\item \textbf{Topological and graph realizations.}
For $A=\delta_k$, emergent cohomology is a shortened coboundary code, Betti
increments are dual-matroid nullities, and the rank-only boundary is
Tutte-determined.  Profile bounds, random-erasure identities, exact simplex
and complete-graph families, and bounded-degree expansion estimates follow.
For every $s\ge2$, connected graphs with identical labeled cut spaces can
still have different $\Shat_{1,s}$ values.
\end{enumerate}

\paragraph{Relation to adjacent work.}
Table~\ref{tab:related-work} separates the data optimized by the closest
coding and combinatorial frameworks.

\begin{table}[tbp]
\centering
\caption{Comparison with adjacent frameworks.}
\label{tab:related-work}
\scriptsize
\renewcommand{\arraystretch}{1.08}
\begin{tabularx}{\textwidth}{@{}P{0.235\textwidth}YY@{}}
\toprule
\textbf{Framework} & \textbf{Recorded quantity} & \textbf{Distinction here}\\
\midrule
Generalized weights, support matroids, and support distributions
\cite{greene1976,klove1992,barg1997,britz2007,wei1991,forney1994,simonis1994,
kurihara2012,zhuang2014,luo2010,luoliu2023,wen2025}
& Support, effective length, or rank data of subcodes and nested-code
quotients.
& Distance is imposed in the domain quotient, while support is charged after
a fixed map in check coordinates.\\

Code distances and test-family generalized weights
\cite{campsmoreno2025,digiusto2025}
& Minimum distances of subcodes or supercodes, or intersection dimensions
against a selected family of test spaces.
& These are invariants of a code with its ambient metric or test family;
$\Shat_{q,s}(A)$ may change under row-equivalent check realizations of the
same code.\\

Generalized covering radii and covering dimension
\cite{elimelech2021,alfarano2026,britzrutherford2005,britzshiromoto2016}
& The number of parity-check columns needed to span prescribed syndromes and
related finite-geometric covering data.
& A feasible hard syndrome subspace forces
$R_q(\ker A)\ge\mathsf N_2(q,s)$, while the objective separately minimizes
its common check support.\\

Stopping sets and stopping redundancy \cite{schwartzvardy2006}
& Matrix-dependent vulnerability of iterative decoding to erasure of variable
coordinates.
& The present erasure acts on listed check generators and constrains every
nonzero released syndrome.\\

Cofilling, graph cuts, and high-dimensional expansion
\cite{gromov2010,dotterrer-kahle2012,matousekwagner2014,kaufmantessler2021,
oppenheimvalentiner2025,fiedler1973,gareyjohnsonstockmeyer1976}
& Minimum fillings, one-defect expansion, or balanced cut costs.
& One erased check set must support a prescribed-dimensional subspace whose
all nonzero combinations satisfy the same localization threshold.\\
\bottomrule
\end{tabularx}
\end{table}

The code distances of Camps-Moreno, Gorla, and L\'opez optimize minimum
distance over subcodes and supercodes of a fixed code \cite{campsmoreno2025}.
The test-family framework of Di Giusto, Gorla, and Ravagnani defines
generalized weights from intersections with selected subspaces of one ambient
metric space \cite{digiusto2025}.  Cofilling shattering is different in kind:
the admissible subspace lies in $\im A$, its distance is pulled back from
$\F^n/\ker A$, and its cost is Hamming support in $\F^m$.  It can be encoded
in a larger object containing the graph of $A$, but then the map and both
coordinate systems are part of the data.  The check-basis theorem in
Section~\ref{sec:examples} makes this dependence unavoidable.

Generalized covering radii are the closest one-space obstruction.  For a
full-rank parity-check matrix $H$, Elimelech, Firer, and Schwartz define
$R_t$ through the smallest column set spanning prescribed syndromes
\cite[Definition~1]{elimelech2021}; Alfarano, Marino, Neri, and Trombetti give
a finite-geometric formulation \cite{alfarano2026}.  Our conditional bound
$R_q(\ker A)\ge\mathsf N_2(q,s)$ does not resolve the proposed unconditional
packing--covering inequality from EFS Section~VI.  It uses the extra witness
of a syndrome subspace whose every nonzero element has coset-leader weight at
least $s$.

The term ``cofilling'' refers to the minimum-preimage coordinate.
``Shattering'' refers to loss of a common set of check generators and is
unrelated to VC dimension.  Sections~\ref{sec:preliminaries}--
\ref{sec:shattering} develop the map formalism and simplicial specialization.
The remaining sections treat graph, weighted, random, and extremal
consequences; the appendices contain the longer calculations.

\section{Check realizations and simplicial specialization}\label{sec:preliminaries}

The coordinate bases are part of a check realization: domain coordinates label
variables or lower-dimensional faces, and codomain coordinates label the checks that
may be erased.

\begin{definition}[Check realization and syndrome localization]
\label{def:check-realization}
Let $A:\F^n\to\F^m$ be a binary linear map with fixed coordinate bases.  Put
$C_A=\ker A$ and $\cS_A=\im A$.  For $y\in\cS_A$, define
\begin{equation}\label{eq:map-lambda}
 \lambda_A(y)=\min\set{\abs{x}:x\in\F^n,\ Ax=y},
 \qquad \lambda_A(0)=0.
\end{equation}
If $Ax=y$, then $\lambda_A(y)=\dist(x,C_A)$.
\end{definition}

For $U\le\F^m$, write $\supp U=\bigcup_{u\in U}\supp(u)$.  Define
\begin{equation}\label{eq:map-profiles}
 \Phi_A(s)=\min\set{\abs y:0\ne y\in\cS_A,\ \lambda_A(y)=s},
 \qquad
 \Sigma_A(s)=\min_{t\ge s}\Phi_A(t),
\end{equation}
with the minimum of an empty set equal to $+\infty$.

\begin{definition}[Map-coupled cofilling shattering]\label{def:map-shat}
For integers $q,s\ge1$, define
\begin{equation}\label{eq:map-shat}
 \Shat_{q,s}(A)=
 \min\set{\abs{\supp U}:U\le\cS_A,\ \dim U=q,\
 \lambda_A(y)\ge s\ \text{for every }0\ne y\in U},
\end{equation}
with value $+\infty$ if no such subspace exists.
\end{definition}

\begin{definition}[Released quotient under check erasure]
\label{def:released-quotient}
For $F\subseteq[m]$, let $\pi_{\bar F}:\F^m\to\F^{[m]\setminus F}$ be
coordinate projection, put $A_{\bar F}=\pi_{\bar F}A$, and define
\[
 \cR_F(A)=\ker A_{\bar F}/\ker A,
 \qquad
 \cS_A[F]=\set{y\in\cS_A:\supp(y)\subseteq F}.
\]
The quotient metric on $\cR_F(A)$ is
$\overline\lambda_A(x+C_A)=\dist(x,C_A)$.
\end{definition}

\begin{theorem}[Exact check-erasure dictionary]
\label{thm:general-erasure-dictionary}
For every $F\subseteq[m]$, the map $x+C_A\mapsto Ax$ induces a canonical
isometric vector-space isomorphism
\begin{equation}\label{eq:general-erasure-dictionary}
 \cR_F(A)\cong\cS_A[F],
 \qquad
 \overline\lambda_A(x+C_A)=\lambda_A(Ax).
\end{equation}
Consequently, $\Shat_{q,s}(A)$ is the least $|F|$ for which the released
quotient $\cR_F(A)$ contains a $q$-dimensional subspace whose every nonzero
class has quotient weight at least $s$.
\end{theorem}

\begin{proof}
A vector $x$ lies in $\ker A_{\bar F}$ exactly when $Ax$ is supported in
$F$.  Hence the restriction $A:\ker A_{\bar F}\to\cS_A[F]$ is surjective
and has kernel $C_A$.  The first isomorphism theorem gives the vector-space
isomorphism.  If $Ax=y$, then the minimum weight in the coset $x+C_A$ is
$\lambda_A(y)$, proving the metric identity.  Finally, a subspace
$U\le\cS_A$ is contained in $\cS_A[F]$ exactly when
$\supp U\subseteq F$; minimizing first over $F$ and then over $U$ gives the
last assertion.
\end{proof}

\begin{proposition}[Natural coordinate equivalence and boundary cases]
\label{prop:map-equivalence}
Let $P$ and $Q$ be permutation matrices on the codomain and domain and put
$B=PAQ$.  Then
\[
 \lambda_B(Py)=\lambda_A(y),\qquad
 \Shat_{q,s}(B)=\Shat_{q,s}(A).
\]
For $1\le q\le\rank A$,
\begin{equation}\label{eq:map-boundaries}
 \Shat_{q,1}(A)=d_q(\cS_A),
 \qquad
 \Shat_{1,s}(A)=\Sigma_A(s),
\end{equation}
and $\Shat_{q,s}(A)$ is nondecreasing in both parameters.
\end{proposition}

\begin{proof}
Permutation matrices preserve Hamming weight.  If $Bx=Py$, then $A(Qx)=y$,
so $x\mapsto Qx$ proves the localization identity.  The map $U\mapsto PU$
preserves dimension, common support, and localization.  The boundary formulas
follow directly from the definitions, and monotonicity follows by restricting
the feasible set.
\end{proof}

Arbitrary row operations are intentionally absent: they replace the listed
check generators by linear combinations and can change the charged support.
Theorem~\ref{thm:check-basis-separation} proves this for row-equivalent
parity-check matrices of the same code.

We now specialize to top-face erasure.  Let $X$ be a finite simplicial complex of dimension exactly $k+1$, and let $X(j)$ denote
its set of $j$-faces.  All chains and cochains use coefficients in $\F$; standard cohomological,
coding-theoretic, and matroid terminology follows
\cite{hatcher2002,macwilliamssloane1977,oxley2011}.  We identify a
cochain with its support, so addition is symmetric difference and the unweighted norm is
Hamming weight:
\[
 C^j(X)=\F^{X(j)},\qquad \abs{\alpha}=\abs{\supp(\alpha)}.
\]
The coboundary maps are
\[
 \delta_j:C^j(X)\longrightarrow C^{j+1}(X),
\]
and
\[
 Z^j(X)=\ker\delta_j,\qquad
 B^j(X)=\im\delta_{j-1},\qquad
 H^j(X)=Z^j(X)/B^j(X).
\]
We write $\beta_j(X)=\dim H^j(X)$.  Since the coefficient field has characteristic two,
orientations play no role.

For $F\subseteq X(k+1)$, let $X-F$ be the subcomplex obtained by deleting precisely the
top faces in $F$ and retaining every face of dimension at most $k$.  Thus
\[
 C^j(X-F)=C^j(X)\quad (j\le k),
 \qquad
 C^{k+1}(X-F)=\F^{X(k+1)\setminus F}.
\]
The map $\delta_{k-1}$ is unchanged, whereas $\delta_k$ is obtained by deleting the rows
indexed by $F$.  In particular,
\begin{equation}\label{eq:B-unchanged}
 B^k(X-F)=B^k(X),
 \qquad
 Z^k(X)\subseteq Z^k(X-F).
\end{equation}

\begin{definition}[Top coboundary code]\label{def:top-code}
The top coboundary code of $X$ in degree $k$ is
\[
 \cC_X^{k+1}=B^{k+1}(X)=\im\delta_k
 \subseteq \F^{X(k+1)}.
\]
For a code $\cC\subseteq\F^E$ and $F\subseteq E$, its shortening to $F$ is
\[
 \cC[F]=\set{c\in\cC:\supp(c)\subseteq F}.
\]
\end{definition}

For $A=\delta_k$, write
\begin{equation}\label{eq:simplicial-map-specialization}
 \lambda_X^k=\lambda_{\delta_k},\quad
 \Phi_X^k=\Phi_{\delta_k},\quad
 \Sigma_X^k=\Sigma_{\delta_k},\quad
 \Shat_{q,s}^k(X)=\Shat_{q,s}(\delta_k).
\end{equation}

We use normalized weights only in Section~\ref{sec:weighted}.  There, $X$ is pure and a
probability measure $\pi$ is fixed on $X(k+1)$.  The induced probability measure on
$k$-faces is
\begin{equation}\label{eq:induced-weight}
 w_k(\sigma)=\frac{1}{k+2}\sum_{\tau\supset\sigma}\pi(\tau).
\end{equation}
For a cochain or face set $A$, write
\[
 \norm{A}_j=\sum_{\sigma\in A}w_j(\sigma),
 \qquad w_{k+1}=\pi.
\]
The factor in \eqref{eq:induced-weight} makes $w_k$ a probability measure.

\begin{table}[tbp]
\centering
\caption{Nonstandard notation used throughout the paper.}
\label{tab:notation}
\footnotesize
\renewcommand{\arraystretch}{1.08}
\begin{tabularx}{\textwidth}{@{}P{0.145\textwidth}YP{0.145\textwidth}Y@{}}
\toprule
\textbf{Symbol} & \textbf{Meaning} & \textbf{Symbol} & \textbf{Meaning}\\
\midrule
$A:\F^n\to\F^m$ & Fixed binary check realization.
& $C_A,\cS_A$ & Kernel code $\ker A$ and syndrome space $\im A$.\\
$\cR_F(A)$ & Quotient released after deleting check coordinates $F$.
& $\cS_A[F]$ & Shortening of the syndrome space to $F$.\\
$\lambda_A(y)$ & Minimum variable weight realizing syndrome $y$.
& $\Shat_{q,s}(A)$ & Minimum check support of a hard $q$-syndrome space.\\
$\cC_X^{k+1}$ & Top coboundary code $\im\delta_k$.
& $\cE_F^k(X)$ & Emergent quotient $H^k(X-F)/H^k(X)$.\\
$\Delta\beta_k(F)$ & Cohomology increment after erasing $F$.
& $\eta_q^k(X)$ & Rank-only threshold for creating $q$ classes.\\
$\lambda_X^k(c)$ & Minimum $k$-cochain size among preimages of $c$.
& $\Phi_X^k(s)$ & Minimum image support at exact localization $s$.\\
$\Sigma_X^k(s)$ & Monotone envelope $\min_{t\ge s}\Phi_X^k(t)$.
& $\Shat_{q,s}^k(X)$ & Simplicial specialization $\Shat_{q,s}(\delta_k)$.\\
$\mathsf N_2(q,s)$ & Shortest binary length for dimension $q$ and distance at least $s$.
& $R_t(C)$ & $t$th generalized covering radius of a code $C$.\\
$\mathsf G_q(d)$ & Binary Griesmer function $\sum_{i=0}^{q-1}\ceil{d/2^i}$.
& $h^k(X)$ & Unnormalized coboundary-expansion constant.\\
$n,m,r$ & $|X(k)|$, $|X(k+1)|$, and $\rank\delta_k$.
& $D$ & Maximum number of top faces containing one $k$-face.\\
$\pi,w_k$ & Top-face probability measure and induced $k$-face measure.
& $\norm{A}_j$ & Weighted mass of a $j$-face set or cochain.\\
\bottomrule
\end{tabularx}
\end{table}

Table~\ref{tab:notation} is also a dependency guide: $\cE_F^k$ and $\Delta\beta_k$ are
outputs of erasure; $\eta_q^k$ records only their dimension; and $\lambda$, $\Phi$,
$\Sigma$, and $\Shat$ progressively restore geometric information about representatives.

\section{Top-face erasures and shortened coboundary codes}\label{sec:exact}

Theorem~\ref{thm:general-erasure-dictionary} applies to every row-deleted
map.  For a last coboundary map it also has a cohomological interpretation,
because top-face deletion leaves the lower coboundary space unchanged.  The
rank-only theory follows from this specialization.

The inclusion $Z^k(X)\subseteq Z^k(X-F)$ and
\eqref{eq:B-unchanged} induce an injection $H^k(X)\hookrightarrow H^k(X-F)$.
This makes the following quotient canonical.

\begin{definition}[Emergent cohomology]\label{def:emergent}
For $F\subseteq X(k+1)$, define
\[
 \cE_F^k(X)=H^k(X-F)/H^k(X).
\]
Its dimension is the cohomology increment
\[
 \Delta\beta_k(F)=\beta_k(X-F)-\beta_k(X).
\]
\end{definition}

Let
\[
 K_F=\set{\alpha\in C^k(X):\supp(\delta_k\alpha)\subseteq F}.
\]
Because the surviving rows of $\delta_k$ vanish on $K_F$, one has
$K_F=Z^k(X-F)$.

\begin{theorem}[Emergent cohomology is a shortened coboundary code]
\label{thm:exact-isomorphism}
For every $F\subseteq X(k+1)$, the coboundary map induces a canonical vector-space
isomorphism
\begin{equation}\label{eq:exact-isomorphism}
 \cE_F^k(X)\cong \cC_X^{k+1}[F].
\end{equation}
Consequently,
\begin{equation}\label{eq:beta-shortening}
 \Delta\beta_k(F)=\dim\cC_X^{k+1}[F].
\end{equation}
\end{theorem}

\begin{proof}
The restriction $\delta_k:K_F\to\cC_X^{k+1}[F]$ is surjective by the definitions of
$K_F$, the image code, and shortening.  Its kernel is $Z^k(X)$.  The first isomorphism theorem gives
\[
 K_F/Z^k(X)\cong \cC_X^{k+1}[F].
\]
On the other hand, since $B^k(X-F)=B^k(X)$,
\[
 \cE_F^k(X)
 =\frac{K_F/B^k(X)}{Z^k(X)/B^k(X)}
 \cong K_F/Z^k(X).
\]
Composing the two canonical maps proves \eqref{eq:exact-isomorphism}, and taking
dimensions proves \eqref{eq:beta-shortening}.
\end{proof}

\begin{remark}[Relative-cohomology interpretation]\label{rem:relative}
The same isomorphism follows from the long exact sequence of the pair $(X,X-F)$.
Because the two complexes have the same $k$-skeleton,
$C^j(X,X-F)=0$ for $j\le k$ and
$H^{k+1}(X,X-F)=C^{k+1}(X,X-F)\cong\F^F$.
The connecting map identifies $\cE_F^k(X)$ with the kernel of
$H^{k+1}(X,X-F)\to H^{k+1}(X)$.  That kernel consists precisely of top cochains
supported in $F$ that are global coboundaries, namely $\cC_X^{k+1}[F]$.
\end{remark}

Let $E=X(k+1)$ and let $M_X^k$ be the binary row matroid of $\delta_k$---equivalently,
the support matroid represented by a generator matrix of $\cC_X^{k+1}$---with ground
set $E$ and rank function $r(A)$ equal to the rank of the rows indexed by $A$.  Let
$(M_X^k)^*$ be its dual and write
$\nu_{(M_X^k)^*}(F)=\abs F-r_{(M_X^k)^*}(F)$ for dual nullity.

\begin{corollary}[Matroid rank defect]\label{cor:rank-defect}
For every $F\subseteq E$,
\begin{equation}\label{eq:rank-defect}
 \Delta\beta_k(F)=r(E)-r(E\setminus F).
\end{equation}
\end{corollary}

\begin{proof}
The dimension of the shortened image code equals the rank lost after restricting all
codewords to $E\setminus F$, which is the row-rank difference on the right side.
Equivalently, compute $\beta_k$ from adjacent coboundary ranks; only the rank of
$\delta_k$ changes.
\end{proof}

\begin{corollary}[Dual-matroid nullity and supermodularity]
\label{cor:submodularity}
For every $F\subseteq E$,
\begin{equation}\label{eq:dual-nullity}
 \Delta\beta_k(F)
 =r(E)-r(E\setminus F)
 =\abs F-r_{(M_X^k)^*}(F)
 =\nu_{(M_X^k)^*}(F).
\end{equation}
Consequently, $F\mapsto\Delta\beta_k(F)$ is monotone nondecreasing and supermodular:
\[
 \Delta\beta_k(A)+\Delta\beta_k(B)
 \le \Delta\beta_k(A\cup B)+\Delta\beta_k(A\cap B).
\]
\end{corollary}

\begin{proof}
The dual-rank identity gives
$r_{(M_X^k)^*}(F)=\abs F+r(E\setminus F)-r(E)$, proving
\eqref{eq:dual-nullity}.  Nullity is monotone and supermodular because matroid rank is
monotone and submodular.
\end{proof}

Thus the Betti increment is a dual-matroid nullity with a topological
interpretation.  Its monotonicity, increasing marginal effect, and Tutte generating
law are matroidal consequences.  The localization-sensitive results below also use
the quotient metric on cochains.

\section{Rank-only vulnerability: generalized Hamming weights}\label{sec:rank}

The correspondence reduces the number of newly created classes to the dimension of
a shortened code.  This gives the rank-only answer, but not the geometry of
representatives.

\begin{definition}[Rank-erasure weights]\label{def:eta}
For $1\le q\le \rank\delta_k$, define
\begin{equation}\label{eq:eta}
 \eta_q^k(X)=\min\set{\abs{F}:F\subseteq X(k+1),\ \Delta\beta_k(F)\ge q}.
\end{equation}
\end{definition}

For a linear code $\cC\subseteq\F^E$, the $q$th generalized Hamming weight is
\[
 d_q(\cC)=\min\set{\abs{\supp U}:U\le\cC,\ \dim U=q},
\]
where $\supp U=\bigcup_{c\in U}\supp(c)$ \cite{wei1991}.

\begin{theorem}[Rank erasures are generalized Hamming weights]
\label{thm:eta-generalized-weight}
For $1\le q\le\rank\delta_k$,
\begin{equation}\label{eq:eta-dq}
 \eta_q^k(X)=d_q(\cC_X^{k+1}).
\end{equation}
\end{theorem}

\begin{proof}
By Theorem~\ref{thm:exact-isomorphism}, $\Delta\beta_k(F)\ge q$ if and only if
$\cC_X^{k+1}[F]$ contains a $q$-dimensional subcode.  This occurs if and only if there
is a $q$-dimensional subcode $U\le\cC_X^{k+1}$ with $\supp U\subseteq F$.  Minimizing
$\abs{F}$ gives exactly $d_q$.
\end{proof}

Standard consequences of generalized-weight theory now become exact topological erasure
statements.  The coding and matroid ingredients in this rank-only layer are classical:
Greene related weight enumerators to Tutte polynomials; Kl{\o}ve developed a
contemporaneous support-weight-distribution viewpoint, and Barg and Britz developed
support matroids for subcodes \cite{greene1976,klove1992,barg1997,britz2007,wei1991}.
The contribution here is the canonical
topological identification of top-face erasure with shortening; the genuinely
map-sensitive invariant begins in Section~\ref{sec:shattering}.

\begin{corollary}[Basic rank-erasure bounds]\label{cor:eta-basic}
Let $m=\abs{X(k+1)}$ and $r=\rank\delta_k$.  Then
\[
 1\le \eta_1^k(X)<\eta_2^k(X)<\cdots<\eta_r^k(X)\le m
\]
and
\begin{equation}\label{eq:generalized-singleton}
 q\le \eta_q^k(X)\le m-r+q.
\end{equation}
\end{corollary}

\begin{remark}[Duality and top homology]\label{rem:dual}
Since $X$ has no $(k+2)$-faces,
\[
 (\cC_X^{k+1})^\perp=\ker\partial_{k+1}=Z_{k+1}(X)=H_{k+1}(X;\F).
\]
Thus Wei duality expresses the erasure hierarchy in terms of the top-homology code
\cite{wei1991}.  In particular, if $H_{k+1}(X;\F)=0$, then
$\cC_X^{k+1}=\F^{X(k+1)}$ and $\eta_q^k(X)=q$.  At the opposite elementary extreme,
if the top-cycle space is spanned by the all-one chain, then the top coboundary code is
the even-weight code and $\eta_q^k(X)=q+1$ for $1\le q<m$.
\end{remark}

Remark~\ref{rem:dual} also shows why rank-erasure weights can be too coarse.  Every
connected closed mod-two pseudomanifold with one-dimensional top homology has the same
rank-erasure hierarchy $q+1$, although its representatives may have very different
geometry.

\section{From rank to localization: the cofilling shattering hierarchy}
\label{sec:shattering}

Rank counts defects but does not distinguish local from global ones.  The quotient
distance below adds that information and leads to the two-parameter hierarchy.

\subsection{Minimum preimages and cofilling profiles}

\begin{definition}[Simplicial localization notation]\label{def:localization}
For $c\in\cC_X^{k+1}$,
\begin{equation}\label{eq:lambda}
 \lambda_X^k(c)=\lambda_{\delta_k}(c)
 =\min\set{\abs{\alpha}:\delta_k\alpha=c}.
\end{equation}
For $c=0$, set $\lambda_X^k(0)=0$.  Define
\begin{align}
 \Phi_X^k(s)&=\min\set{\abs c:0\ne c\in\cC_X^{k+1},\ \lambda_X^k(c)=s},
 \label{eq:Phi}\\
 \Sigma_X^k(s)&=\min_{t\ge s}\Phi_X^k(t).
 \label{eq:Sigma}
\end{align}
\end{definition}

If $c=\delta_k\alpha$, then
\begin{equation}\label{eq:lambda-distance}
 \lambda_X^k(c)=\dist(\alpha,Z^k(X))
 =\min_{z\in Z^k(X)}\abs{\alpha+z}.
\end{equation}
Thus localization is the coset-leader weight of $c$ for $Z^k(X)$.

The unnormalized coboundary-expansion constant is
\begin{equation}\label{eq:h-unweighted}
 h^k(X)=\min_{\alpha\notin Z^k(X)}
 \frac{\abs{\delta_k\alpha}}{\dist(\alpha,Z^k(X))}.
\end{equation}
Equations \eqref{eq:lambda-distance} and \eqref{eq:Phi} give the exact relation
\begin{equation}\label{eq:h-profile}
 h^k(X)=\min_{s:\Phi_X^k(s)<\infty}\frac{\Phi_X^k(s)}{s}.
\end{equation}
Thus $h^k$ is the lower supporting slope of the full localization profile.

\begin{proposition}[Dimension zero recovers graph isoperimetry]
\label{prop:graph-profile}
Let $G$ be a connected graph, viewed as a one-dimensional simplicial complex, and put
$k=0$.  For $1\le s\le\floor{\abs{V(G)}/2}$,
\begin{equation}\label{eq:graph-isoperimetry}
 \Phi_G^0(s)=\min_{\substack{S\subseteq V(G)\\\abs{S}=s}}\abs{\partial_G S}.
\end{equation}
Consequently, exact evaluation of the middle localization profile is NP-hard.
\end{proposition}

\begin{proof}
A zero-cochain is the indicator of a vertex set $S$, and its coboundary is the cut
$\partial_G S$.  Since $G$ is connected, $Z^0(G)$ consists of the two constant
cochains.  Hence the distance of $\1_S$ to $Z^0(G)$ is
$\min\{\abs S,\abs{V\setminus S}\}$.  For $s$ at most half the order, minimizing the
cut support over coset leaders of weight $s$ is exactly \eqref{eq:graph-isoperimetry}.
For even order and $s=\abs V/2$, this is minimum bisection, which is NP-hard
\cite{gareyjohnsonstockmeyer1976}.
\end{proof}

\subsection{The two-parameter hierarchy}

Through Theorem~\ref{thm:exact-isomorphism}, every $u\in\cE_F^k(X)$ corresponds to a
unique codeword $c_u\in\cC_X^{k+1}[F]$.  Define its relative localization by
\[
 \ell_X^k(u)=\lambda_X^k(c_u).
\]
The word ``relative'' is important: $u$ is measured modulo the old cocycle space, not
against the zero cochain alone.

Viewed as a parity-check map for the cocycle code $Z^k(X)\subseteq C^k(X)$,
$\delta_k$ is its syndrome map and $\lambda_X^k$ is syndrome coset-leader weight.
The definition therefore asks for a high-distance syndrome subspace with small common
support.

\begin{definition}[Cofilling shattering of a complex]\label{def:shat}
For integers $q,s\ge1$, set
$\Shat_{q,s}^k(X)=\Shat_{q,s}(\delta_k)$.  Equivalently, it is the least
$\abs F$ for which a $q$-dimensional subspace $U\le\cE_F^k(X)$ satisfies
$\ell_X^k(u)\ge s$ for every $0\ne u\in U$.
\end{definition}

\begin{theorem}[Erasure and code formulations]\label{thm:shat-code}
For all $q,s\ge1$,
\begin{equation}\label{eq:shat-code}
 \Shat_{q,s}^k(X)=
 \min\set{\abs{\supp U}:
 U\le\cC_X^{k+1},\ \dim U=q,\
 \lambda_X^k(c)\ge s\ \text{for all }0\ne c\in U}.
\end{equation}
In particular,
\begin{equation}\label{eq:shat-boundaries}
 \Shat_{q,1}^k(X)=\eta_q^k(X)=d_q(\cC_X^{k+1}),
 \qquad \Shat_{1,s}^k(X)=\Sigma_X^k(s).
\end{equation}
\end{theorem}

\begin{proof}
Apply Definition~\ref{def:map-shat} to $A=\delta_k$.  Under
Theorem~\ref{thm:exact-isomorphism}, a subspace of $\cE_F^k(X)$ is a subcode
of $\cC_X^{k+1}[F]$, and relative localization becomes $\lambda_X^k$.
Choosing $F=\supp U$ proves the equality.  The boundary formulas follow from
\eqref{eq:map-boundaries} and Theorem~\ref{thm:eta-generalized-weight}.
\end{proof}

\begin{remark}[Why every nonzero combination is constrained]
\label{rem:all-combinations}
Requiring only a basis of well-localized classes would depend on the basis and
would allow difficult syndromes to cancel into an easy one.  The condition on
every nonzero element is the analogue of minimum distance for a linear code.
\end{remark}

\subsection{Domain feasibility and a profile-sensitive Griesmer bound}

For integers $q,s\ge1$, let
\begin{equation}\label{eq:shortest-code-length}
 \mathsf N_2(q,s)=
 \min\set{\ell:\text{there exists a binary linear $[\ell,q,d]$ code with $d\ge s$}}
\end{equation}
and write
\begin{equation}\label{eq:griesmer-function}
 \mathsf G_q(d)=\sum_{i=0}^{q-1}\ceil{\frac{d}{2^i}}.
\end{equation}
Thus $\mathsf N_2(q,s)\ge\mathsf G_q(s)$ by the binary Griesmer bound \cite{griesmer1960}.

For a binary linear code $C\subseteq\F^n$ with full-rank parity-check matrix
$H=[h_1\ \cdots\ h_n]$, Elimelech, Firer, and Schwartz define the generalized
covering radius by the following parity-check-column formulation
\cite[Definition~1]{elimelech2021}:
\begin{equation}\label{eq:generalized-covering-radius}
 R_t(C)=
 \max_{\substack{S\subseteq\F^{n-\dim C}\\|S|=t}}
 \min\set{\abs I:I\subseteq[n],\ S\subseteq
 \operatorname{span}\set{h_i:i\in I}}.
\end{equation}
Their Lemma~2 shows that replacing $H$ by $AH$ with
$A\in\operatorname{GL}(n-\dim C,2)$ leaves $R_t(C)$ unchanged.  Their hierarchy
satisfies $R_t(C)\le n-\dim C$ and Proposition~13 gives
$R_t(C)\le tR_1(C)$ \cite[Lemma~2, Eq.~(2), and Proposition~13]{elimelech2021}.
The localization and support requirements live in different coordinate spaces, but a
feasible syndrome subspace imposes a nontrivial obstruction on this domain-side
covering parameter.

\begin{theorem}[Rank and generalized-covering-radius obstruction]
\label{thm:domain-code}
Let $A:\F^n\to\F^m$, put $r=\rank A$, and let $C_A=\ker A$.  If
$\Shat_{q,s}(A)<\infty$, then $q\le r$ and
\begin{equation}\label{eq:domain-covering}
 r\ge R_q(C_A)\ge \mathsf N_2(q,s)\ge \mathsf G_q(s).
\end{equation}
Moreover,
\begin{equation}\label{eq:ordinary-covering-obstruction}
 R_1(C_A)\ge\ceil{\frac{\mathsf N_2(q,s)}{q}},
\end{equation}
and there exists a binary $[r,q,d]$ code with $d\ge s$, so
\begin{equation}\label{eq:domain-hamming}
 2^q\sum_{j=0}^{\floor{(s-1)/2}}\binom rj\le2^r.
\end{equation}
For $A=\delta_k$, these are obstructions for $\Shat_{q,s}^k(X)$.
\end{theorem}

\begin{proof}
Choose an isomorphism $\psi:\im A\to\F^r$ and take $H=\psi A$ as a
full-rank parity-check matrix of $C_A$.  EFS Lemma~2 makes $R_q(C_A)$
independent of this choice \cite[Lemma~2]{elimelech2021}.

Let $U\le\im A$ be feasible with basis $u_1,\ldots,u_q$.  If $I\subseteq[n]$
is any coordinate set whose columns of $H$ span
$\psi(u_1),\ldots,\psi(u_q)$, then $U\subseteq A(\F^I)$.  The restriction
$A:\F^I\cap A^{-1}(U)\to U$ is surjective.  Choose a complement $V$ to its
kernel $\F^I\cap C_A$.  Then $\dim V=q$, $A|_V$ is an isomorphism onto $U$,
and every $0\ne v\in V$ has $\abs v\ge\lambda_A(Av)\ge s$.  After
puncturing zero coordinates, $V$ is a binary code of length at most $\abs I$,
dimension $q$, and distance at least $s$.  Thus $\abs I\ge\mathsf N_2(q,s)$,
and Definition~1 of EFS gives $R_q(C_A)\ge\mathsf N_2(q,s)$.  The Griesmer
bound gives the final inequality in \eqref{eq:domain-covering}.

Feasibility gives $q\le r$.  EFS Eq.~(2) gives $R_q(C_A)\le R_r(C_A)=r$,
and Proposition~13 gives $R_q(C_A)\le qR_1(C_A)$
\cite[Eq.~(2) and Proposition~13]{elimelech2021}.  Finally, apply the same
complement construction to $r$ independent columns of $H$ to obtain a binary
$[r,q,d]$ code with $d\ge s$; the sphere-packing bound gives
\eqref{eq:domain-hamming} \cite{macwilliamssloane1977}.
\end{proof}

\begin{remark}[Relation to generalized packing radii]
Elimelech--Firer--Schwartz define
$\delta_t(C)=\floor{(d_t(C)-1)/2}$ and prove a rank-restricted separation result for
$t$-balls \cite[Section~VI and Lemma~27]{elimelech2021}.  They observe that
$\delta_t(C)\le R_t(C)$ would extend the ordinary packing--covering inequality, but the
rank restriction prevents their lemma from yielding a packing of all $t$-balls when
$t\ge2$.  Theorem~\ref{thm:domain-code} neither assumes nor proves that inequality.  It
uses the additional map-coupled witness $U$ to manufacture a $q$-dimensional domain
code of distance at least $s$ inside every coordinate set spanning a basis of $U$.
Thus its lower bound is conditional and is expressed through the shortest-code-length
function $\mathsf N_2(q,s)$ rather than the generalized packing radius of $C$.
\end{remark}

\begin{remark}[Two independent coordinate-space obstructions]
Theorem~\ref{thm:domain-code} converts the covering-radius comparison into a theorem:
a selected hard syndrome space forces large map rank, large generalized covering
radius, and a lower bound on the ordinary covering radius in domain coordinates.
The next theorem instead constrains the effective length of the image subcode in top-face
coordinates after localization has forced a minimum syndrome weight.  Neither
coordinate-space obstruction contains the other in general.
\end{remark}

\begin{proposition}[Map-coupled profile--Griesmer bound]
\label{prop:map-profile-griesmer}
For every check realization $A$ and every feasible pair $(q,s)$,
\begin{equation}\label{eq:map-profile-griesmer}
 \Shat_{q,s}(A)\ge
 \max\set{d_q(\im A),\mathsf G_q\bigl(\Sigma_A(s)\bigr)}.
\end{equation}
\end{proposition}

\begin{proof}
Puncturing outside $\supp U$ turns a feasible $q$-space into a binary code
of length $\abs{\supp U}$ and distance at least $\Sigma_A(s)$.  Apply the
Griesmer bound and the definition of $d_q(\im A)$, then minimize over $U$.
\end{proof}

\begin{theorem}[Profile-Griesmer lower bound for cofilling shattering]
\label{thm:griesmer-shat}
Assume $\cC_X^{k+1}\ne0$.  For every $q,s\ge1$ for which
$\Shat_{q,s}^k(X)<\infty$,
\begin{align}
 \Shat_{q,s}^k(X)
 &\ge
 \max\set{
 d_q(\cC_X^{k+1}),
 \mathsf G_q\bigl(\Sigma_X^k(s)\bigr)
 }
 \label{eq:profile-griesmer-shat}\\
 &\ge
 \max\set{
 d_q(\cC_X^{k+1}),
 \mathsf G_q\bigl(\ceil{h^k(X)s}\bigr)
 }.
 \label{eq:griesmer-shat}
\end{align}
\end{theorem}

\begin{proof}
Equation~\eqref{eq:profile-griesmer-shat} is
Proposition~\ref{prop:map-profile-griesmer} applied to $A=\delta_k$.
Finally, if $t\ge s$ then
$\Phi_X^k(t)\ge h^k(X)t\ge h^k(X)s$ by \eqref{eq:h-profile}.  Since
$\Phi_X^k(t)$ is integral, $\Sigma_X^k(s)\ge\ceil{h^k(X)s}$.  The function
$\mathsf G_q$ is nondecreasing, proving \eqref{eq:griesmer-shat}.
\end{proof}

The first inequality retains the complete one-defect geometry rather than only its
smallest slope.  The second converts an ordinary expansion constant into a
multi-dimensional erasure bound.  The ceiling structure matters: even when every
nonzero syndrome has weight only two, linearity forces a $q$-dimensional defect space
to occupy at least $q+1$ coordinates.

\subsection{A coarse macroscopic existence window}

The preceding results are lower bounds.  For an upper bound, we first choose a
subcode with small top-face support and then find within it a quotient subspace that
avoids all short vectors.  The resulting density window is qualitative and generally
nonsharp; it establishes a nontrivial linear scale rather than competitive constants
for known constructions.  Write
\[
 n=\abs{X(k)},\qquad m=\abs{X(k+1)},\qquad
 r=\rank\delta_k=\dim\cC_X^{k+1}.
\]
For a $k$-face $\sigma$, its upper degree is
$\deg_X(\sigma)=\abs{\{\tau\in X(k+1):\sigma\subset\tau\}}$.

\begin{lemma}[Rank density from bounded incidence]
\label{lem:rank-density}
Suppose $X$ is pure and $\deg_X(\sigma)\le D$ for every $\sigma\in X(k)$.  Then
\begin{equation}\label{eq:rank-density}
 r\ge \frac{n}{D(k+1)+1}
 \qquad\text{and}\qquad
 n\ge\frac{k+2}{D}\,m.
\end{equation}
Consequently,
\[
 r\ge \frac{k+2}{D\bigl(D(k+1)+1\bigr)}\,m.
\]
\end{lemma}

\begin{proof}
Form a conflict graph on $X(k)$ by joining two distinct $k$-faces when they lie in a
common top face.  Its maximum degree is at most $D(k+1)$, so a greedy independent-set
argument gives an independent set $I$ of size at least $n/(D(k+1)+1)$.  Purity makes
$\delta_k\1_\sigma$ nonzero for every $\sigma\in I$, and independence of $I$ makes the
supports of these coboundaries pairwise disjoint.  They are therefore linearly
independent columns of $\delta_k$, proving the first inequality.  Double-counting
incidences gives
\[
 (k+2)m=\sum_{\sigma\in X(k)}\deg_X(\sigma)\le Dn,
\]
which proves the second inequality and the consequence.
\end{proof}

For integers $n,s\ge1$, put
\[
 \mathsf B(n,s)=\sum_{j=1}^{\min\{s-1,n\}}\binom{n}{j},
\]
with an empty sum equal to zero.  The next theorem is a finite
Gilbert--Varshamov-type avoidance statement inside an optimally supported subcode.  Its
second inequality is the generalized Singleton bound, included for completeness.

\begin{theorem}[Support--localization avoidance]
\label{thm:localized-existence}
Let $1\le q\le t\le r$ and $s\ge1$.  If
\begin{equation}\label{eq:localized-existence-condition}
 \mathsf B(n,s)\,(2^q-1)<2^t-1,
\end{equation}
then
\begin{equation}\label{eq:localized-existence-conclusion}
 \Shat_{q,s}^k(X)
 \le d_t(\cC_X^{k+1})
 \le m-r+t.
\end{equation}
Equivalently, there is a $q$-dimensional subspace of $C^k(X)/Z^k(X)$ whose nonzero
cosets all have distance at least $s$ from $Z^k(X)$ and whose coboundaries have common
support of size at most $d_t(\cC_X^{k+1})$.
\end{theorem}

\begin{proof}[Proof idea]
Choose a $t$-dimensional subcode $W$ attaining $d_t(\cC_X^{k+1})$ and transport it
through the quotient isomorphism
$C^k(X)/Z^k(X)\cong\cC_X^{k+1}$.  At most $\mathsf B(n,s)$ nonzero quotient classes in
this $t$-space have norm below $s$.  A uniformly random $q$-subspace contains any fixed
nonzero vector with probability $(2^q-1)/(2^t-1)$, so the stated inequality and a union
bound produce a subspace avoiding every short class.  Its image remains supported inside
$\supp W$.  The final inequality is the generalized Singleton bound obtained from an
information set.  The full counting argument is in
Appendix~\ref{app:localized-existence}.
\end{proof}

\begin{corollary}[Support--localization sandwich]
\label{cor:support-localization-sandwich}
If $q\le t\le r$ satisfies \eqref{eq:localized-existence-condition}, then
\begin{equation}\label{eq:support-localization-sandwich}
 \max\left\{
 d_q(\cC_X^{k+1}),
 \mathsf G_q\bigl(\Sigma_X^k(s)\bigr)
 \right\}
 \le \Shat_{q,s}^k(X)
 \le d_t(\cC_X^{k+1})
 \le m-r+t.
\end{equation}
Minimizing the middle upper bound over admissible $t$ gives a canonical
generalized-weight upper envelope for the hierarchy.
\end{corollary}

For $0\le x\le1$, let
\[
 H_2(x)=-x\log_2x-(1-x)\log_2(1-x),
\]
with the usual endpoint convention.

\begin{theorem}[Coarse macroscopic bounded-degree window]
\label{thm:macroscopic-shattering}
Fix $k\ge0$ and $D\ge1$.  Let $(X_i)$ be a sequence of finite pure
$(k+1)$-complexes such that, with
\[
 n_i=\abs{X_i(k)},\qquad m_i=\abs{X_i(k+1)},
\]
one has $n_i\to\infty$ and every $k$-face has upper degree at most $D$.  Set
\[
 \rho=\frac{1}{D(k+1)+1}.
\]
If $0<\beta<1/2$ and $\alpha>0$ satisfy
\begin{equation}\label{eq:macroscopic-region}
 H_2(\beta)+\alpha<\rho,
\end{equation}
and
\[
 q_i=\floor{\alpha n_i},\qquad s_i=\ceil{\beta n_i},
\]
then $\Shat_{q_i,s_i}^k(X_i)$ is finite for all sufficiently large $i$ and
\begin{equation}\label{eq:macroscopic-density-bounds}
 \frac{k+2}{D}\,\alpha
 \le
 \liminf_{i\to\infty}\frac{\Shat_{q_i,s_i}^k(X_i)}{m_i}
 \le
 \limsup_{i\to\infty}\frac{\Shat_{q_i,s_i}^k(X_i)}{m_i}
 \le
 1-\frac{k+2}{D}\bigl(\rho-H_2(\beta)-\alpha\bigr)
 <1.
\end{equation}
In particular, $\Shat_{q_i,s_i}^k(X_i)=\Theta(m_i)$.  If, in addition,
$h^k(X_i)\ge\varepsilon>0$ uniformly, then the leftmost term in
\eqref{eq:macroscopic-density-bounds} may be replaced by
\begin{equation}\label{eq:macroscopic-expanded-lower}
 \frac{k+2}{D}\max\{\alpha,\varepsilon\beta\}.
\end{equation}
\end{theorem}

\begin{proof}[Proof idea]
Lemma~\ref{lem:rank-density} supplies $r_i\ge\rho n_i$.  Choose an intermediate dimension
$t_i\sim\tau n_i$ with $H_2(\beta)+\alpha<\tau<\rho$.  The entropy estimate for the
Hamming ball makes the avoidance condition in
Theorem~\ref{thm:localized-existence} hold, while its Singleton term and the incidence
comparison $n_i/m_i\ge(k+2)/D$ give the strict upper-density bound.  The lower-density
bound comes from $d_{q_i}\ge q_i$; uniform expansion adds
$\Sigma_X^k(s_i)\ge\varepsilon s_i$.  Appendix~\ref{app:macroscopic-shattering}
contains the limit calculation and endpoint bookkeeping.
\end{proof}

\begin{corollary}[Explicit bounded-degree cosystolic expanders]
\label{cor:explicit-cosystolic}
Suppose $(X_i)$ is a family of pure $(k+1)$-complexes with $|X_i(k)|\to\infty$ and
upper degree at most $D$ that is uniformly $(\varepsilon_0,\mu)$-cosystolically
expanding in degree $k$ under the uniform measure on top faces.  Put
$\rho=1/(D(k+1)+1)$.  Then, for every $\alpha,\beta$ satisfying
\eqref{eq:macroscopic-region},
\begin{align}\label{eq:explicit-density-bounds}
 \frac{k+2}{D}
 \max\left\{\alpha,\frac{\varepsilon_0\beta}{k+2}\right\}
 &\le
 \liminf_{i\to\infty}
 \frac{\Shat_{\floor{\alpha\abs{X_i(k)}},
                   \ceil{\beta\abs{X_i(k)}}}^k(X_i)}{\abs{X_i(k+1)}}
 \notag\\
 &\le
 \limsup_{i\to\infty}
 \frac{\Shat_{\floor{\alpha\abs{X_i(k)}},
                   \ceil{\beta\abs{X_i(k)}}}^k(X_i)}{\abs{X_i(k+1)}}
 \notag\\
 &\le
 1-\frac{k+2}{D}\left(\rho-H_2(\beta)-\alpha\right).
\end{align}
Consequently, every fixed dimension admits explicit families and a nonempty open set of
positive pairs $(\alpha,\beta)$ for which the shattering density is bounded uniformly
away from zero and one.  Evra and Kaufman provide one verified source of such families
\cite{evrakaufman2024}.  In arXiv:1510.00839v3, their Definition~1.5 (equivalently
Definition~2.1) assigns a $j$-face $\sigma$ the weight
\[
 w_X(\sigma)=
 \frac{\deg_X(\sigma)}{\binom{d+1}{j+1}\abs{X(d)}}.
\]
Equivalently, this is the probability that $\sigma$ lies in a uniformly chosen top
$d$-face, divided by $\binom{d+1}{j+1}$.  Thus $j=k$ and $d=k+1$ give
$\deg_X(\sigma)/((k+2)m)$, exactly the balanced measure in
\eqref{eq:induced-weight}.  Their Corollary~6.3 gives cosystolic expansion for the
$d$-dimensional skeletons of suitable $(d+1)$-dimensional Ramanujan complexes; its
Ramanujan-complex hypotheses are verified by Theorem~6.1, and the comparison between
ambient and skeleton norms under bounded degree is quantified in Lemma~2.10.  Taking
$d=k+1$ therefore matches both the dimension and the top-face normalization assumed
here.  Other constructions fall under the corollary only after the stated measure and
upper-degree hypotheses are checked.
\end{corollary}

\begin{proof}
For the uniform top-face measure, a $k$-face $\sigma$ has induced weight
$\deg_X(\sigma)/((k+2)m)$.  Purity therefore gives, for every cochain $\alpha$,
\[
 \dist_{w_k}(\alpha,Z^k(X))
 \ge\frac{1}{(k+2)m}\dist(\alpha,Z^k(X)).
\]
The weighted expansion inequality becomes
\[
 \frac{\abs{\delta_k\alpha}}{m}
 \ge\varepsilon_0\dist_{w_k}(\alpha,Z^k(X))
 \ge\frac{\varepsilon_0}{(k+2)m}\dist(\alpha,Z^k(X)),
\]
so $h^k(X)\ge\varepsilon_0/(k+2)$.  Theorem~\ref{thm:macroscopic-shattering}
applies.  The entropy region is nonempty because $H_2(\beta)\to0$ as
$\beta\downarrow0$.
\end{proof}

\begin{remark}[Quantitative limitation]\label{rem:macroscopic-limitation}
The constants in Theorem~\ref{thm:macroscopic-shattering} are deliberately coarse.  The
rank estimate uses only $r\ge n/(D(k+1)+1)$, and the upper bound ends with generalized
Singleton.  When $D$ is large, the admissible entropy region can force $\beta$ to be very
small and the upper density can be close to one.  The theorem therefore establishes
linear scaling and a nonempty feasibility window; it does not claim numerically strong
resilience constants for currently known high-dimensional expanders.  Improving the
rank density or the generalized-weight upper envelope is a separate problem.
\end{remark}

\begin{remark}[Finite-field extension]\label{rem:qary}
The exact erasure--code dictionary, the rank hierarchy, the syndrome-space definition,
and the profile--Griesmer argument extend over $\mathbb F_Q$ after restoring orientation
signs.  The counting step in Theorem~\ref{thm:localized-existence} is analogous but not
literally verbatim: one replaces
\[
 \mathsf B(n,s)
 \quad\text{by}\quad
 \mathsf B_Q(n,s)=\sum_{j=1}^{\min\{s-1,n\}}\binom nj(Q-1)^j
\]
and the containment probability by $(Q^q-1)/(Q^t-1)$.  The Griesmer function becomes
$\sum_{i=0}^{q-1}\ceil{d/Q^i}$, and the asymptotic estimate uses the $Q$-ary Hamming-ball
entropy.  We retain $\F$ to match the main high-dimensional-expander applications and to
keep the support combinatorics unsigned.
\end{remark}

\section{Graph cuts as Fourier-balanced multiway cuts}
\label{sec:structural}

The graph case makes every component of the definition visible: quotient distance is
the smaller side of a cut, subspace dimension is the affine rank of a label set, and
common erasure support is the union of all character cuts.  A closing remark records a
relative CSS check-loss interpretation without treating it as a post-deletion code
performance parameter.

Let $G=(V,E)$ be a connected graph on $n$ vertices and let
$0=\lambda_1(L_G)\le\lambda_2(L_G)\le\cdots\le\lambda_n(L_G)$ be the eigenvalues of
its unnormalized Laplacian.  For a labeling $x:V\to\F^q$, put
\[
 E_x=\{uv\in E:x_u\ne x_v\},
 \qquad
 S_t(x)=\{v\in V:t\cdot x_v=1\}
 \quad(0\ne t\in\F^q).
\]
Write $\operatorname{affspan}x(V)$ for the affine span of the used labels.

\begin{theorem}[Fourier-balanced multiway-cut formulation]
\label{thm:fourier-balanced-cut}
For every connected graph $G$ and integers $q,s\ge1$,
\begin{equation}\label{eq:graph-labeling-formulation}
 \Shat_{q,s}^0(G)=
 \min\left\{\abs{E_x}:\begin{array}{l}
 x:V\to\F^q,\quad \operatorname{affspan}x(V)=\F^q,\\
 s\le\abs{S_t(x)}\le n-s\quad\text{for every }0\ne t\in\F^q
 \end{array}\right\},
\end{equation}
where the minimum of the empty set is $+\infty$.  Every feasible labeling satisfies
\begin{equation}\label{eq:graph-fourier-energy}
 \abs{E_x}=2^{1-q}\sum_{0\ne t\in\F^q}\abs{\partial_G S_t(x)}.
\end{equation}
Consequently, whenever the shattering parameter is finite,
\begin{equation}\label{eq:graph-spectral-bounds}
 \Shat_{q,s}^0(G)\ge
 \max\left\{
 2(1-2^{-q})\frac{\lambda_2(L_G)}{n}s(n-s),
 \frac12\sum_{i=2}^{q+1}\lambda_i(L_G)
 \right\}.
\end{equation}
\end{theorem}

\begin{proof}[Proof idea]
Choose a basis of the cut subspace and use its coordinate values as labels
$x_v\in\F^q$.  Injectivity of the cut map is equivalent to full affine span of the used
labels, and each nonzero character $t$ produces exactly the cut
$\partial_G S_t(x)$.  An edge lies in the common support precisely when its endpoint
labels differ; double-counting the characters that separate an unequal pair gives
\eqref{eq:graph-fourier-energy}.  The first spectral bound is Fiedler's cut inequality
inserted character by character.  The second applies Ky Fan's principle to the
normalized label-class indicators.  Full details are in
Appendix~\ref{app:fourier-proof}.
\end{proof}

\begin{remark}[What is specific to shattering]
A conventional multiway cut prescribes or penalizes parts.  Here the parts are label
fibers, their number is endogenous, and feasibility is imposed simultaneously on every
nontrivial Walsh character of the label distribution.  This linear-character constraint
is what distinguishes \eqref{eq:graph-labeling-formulation} from an ordinary balanced
partition problem.  The spectral ingredients themselves are classical
\cite{fiedler1973}.
\end{remark}

\begin{remark}[CSS check deletion as a relative quotient]\label{rem:css-interpretation}
The three-term cochain complex
\[
 C^{k-1}(X)\xrightarrow{\delta_{k-1}}C^k(X)
 \xrightarrow{\delta_k}C^{k+1}(X)
\]
defines a CSS code with qubits indexed by $X(k)$, $X$-check matrix
$H_X=\delta_{k-1}^{\mathsf T}$, and $Z$-check matrix $H_Z=\delta_k$
\cite{calderbankshor1996,steane1996,audouxcouvreur2019}.  Deleting the rows indexed by
$F$ increases the number of encoded qubits by $\Delta\beta_k(F)$, and the quotient of
newly released $X$-logical operators by the pre-deletion logical space is
$\cE_F^k(X)\cong\cC_X^{k+1}[F]$.  Under this translation,
$\Shat_{q,s}^k(X)$ is the minimum number of deleted $Z$-check rows needed to release a
$q$-dimensional relative logical quotient whose nonzero elements have distance at least
$s$ from the old cocycle sector.

This is a relative check-loss parameter.  It is neither physical-qubit erasure nor the
absolute minimum distance of the post-deletion CSS code, because localization is taken
modulo the entire pre-deletion cocycle space, including old logical operators.  We record
the correspondence only as an interpretation of the algebraic quotient and make no
claim that it is the quantum-code objective optimized in code surgery
\cite{cowtanburton2024}.
\end{remark}

\section{Weighted and packed resilience}
\label{sec:weighted}

The preceding theory charges every face equally and treats a subspace through its common
support.  This section shows how the same logic behaves under normalized face weights
and under a geometric packing hypothesis on several defects.

We now use the normalized weights of \eqref{eq:induced-weight}.  After top-face
deletion, all norms are measured in the original ambient complex: the induced $k$-face
weights remain fixed, the mass of $F$ is computed with the original top-face measure,
and no renormalization is performed.  For a subspace $A\subseteq C^k(X)$ and a
cochain $\alpha$, set
\[
 \dist_{w_k}(\alpha,A)=\min_{a\in A}\norm{\alpha+a}_k.
\]

\begin{definition}[Cosystolic expansion]\label{def:cosystolic}
A pure $(k+1)$-complex $X$ is an $(\varepsilon,\mu)$-cosystolic expander in degree $k$
if
\begin{align}
 \norm{\delta_k\alpha}_{k+1}
 &\ge \varepsilon\,\dist_{w_k}(\alpha,Z^k(X))
 &&\text{for every }\alpha\in C^k(X),
 \label{eq:cosystolic-expansion-a}\\
 \dist_{w_k}(z,B^k(X))
 &\ge\mu
 &&\text{for every }z\in Z^k(X)\setminus B^k(X).
 \label{eq:cosystolic-expansion-b}
\end{align}
\end{definition}

This is the standard separation between cofilling and cosystolic size
\cite{evrakaufman2024,diksteindinur2024}.  For an emergent class
$u\in\cE_F^k(X)$ with codeword $c_u$, define its weighted relative localization by
\[
 \ell_{X,w}^k(u)=
 \min\set{\norm{\alpha}_k:\delta_k\alpha=c_u}
 =\dist_{w_k}(\alpha,Z^k(X)).
\]

\begin{theorem}[Old-global/new-relative-local resilience]
\label{thm:old-new}
Let $X$ be an $(\varepsilon,\mu)$-cosystolic expander in degree $k$, and let
$F\subseteq X(k+1)$.  Then:
\begin{enumerate}[label=\textup{(\roman*)},leftmargin=2.2em]
\item the injection $H^k(X)\hookrightarrow H^k(X-F)$ preserves the ambient
cosystolic norm, so every old nonzero class still has norm at least $\mu$;
\item every emergent quotient class $u\in\cE_F^k(X)$ satisfies
\begin{equation}\label{eq:new-local}
 \ell_{X,w}^k(u)\le \frac{\norm F_{k+1}}{\varepsilon}.
\end{equation}
\end{enumerate}
\end{theorem}

\begin{proof}
The $k$-skeleton, the induced ambient $k$-face weights, and $B^k(X)$ are unchanged by
top-face deletion.  Hence the norm of an old cohomology class, measured by distance to
$B^k(X)$, is unchanged, proving (i).

For (ii), let $c_u\in\cC_X^{k+1}[F]$ correspond to $u$ and choose any
$\alpha$ with $\delta_k\alpha=c_u$.  Since $c_u$ is supported in $F$,
\[
 \varepsilon\,\ell_{X,w}^k(u)
 \le\norm{c_u}_{k+1}
 \le\norm F_{k+1}
\]
by \eqref{eq:cosystolic-expansion-a}.  Rearrangement proves
\eqref{eq:new-local}.
\end{proof}

\begin{remark}[The quotient qualification is essential]\label{rem:quotient}
The theorem does not say that every absolute class of $H^k(X-F)$ is either old and global
or new and local.  A new absolute class can be shifted by an old class.  What is canonical
is the quotient $H^k(X-F)/H^k(X)$, and it is this relative class that admits a localized
representative.
\end{remark}

Theorem~\ref{thm:old-new} controls one quotient class at a time.  Multiple defects can
share erased top faces, so a direct sum of the one-class bounds is false without a
geometric separation hypothesis.  The following packed version gives the correct local
incidence loss.

\begin{definition}[Packed defect certificate]\label{def:packed}
A family $\alpha_1,\ldots,\alpha_q\in C^k(X)$ is an $(F,\sigma)$-packed defect certificate if
\begin{enumerate}[label=\textup{(\alph*)},leftmargin=2.2em]
\item $\alpha_i$ is reduced modulo cocycles:
$\norm{\alpha_i}_k=\dist_{w_k}(\alpha_i,Z^k(X))$;
\item the supports $\supp(\alpha_i)$ are pairwise disjoint and
$\norm{\alpha_i}_k\ge \sigma$;
\item $0\ne\delta_k\alpha_i$ and $\supp(\delta_k\alpha_i)\subseteq F$.
\end{enumerate}
\end{definition}

\begin{theorem}[Packed shattering bound]\label{thm:packed}
Suppose $X$ satisfies the first cosystolic inequality
\eqref{eq:cosystolic-expansion-a} with constant $\varepsilon$.  If
$\alpha_1,\ldots,\alpha_q$ is an $(F,\sigma)$-packed defect certificate, then
\begin{equation}\label{eq:packed}
 \norm F_{k+1}\ge \frac{\varepsilon q \sigma}{k+2}.
\end{equation}
The factor $k+2$ is sharp for the local overlap argument.
\end{theorem}

\begin{proof}
Expansion and reducedness give
\[
 \sum_{i=1}^q\norm{\delta_k\alpha_i}_{k+1}
 \ge\varepsilon\sum_{i=1}^q\norm{\alpha_i}_k
 \ge\varepsilon q \sigma.
\]
Fix a top face $\tau$.  If $(\delta_k\alpha_i)(\tau)=1$, then the support of
$\alpha_i$ contains at least one of the $k+2$ $k$-faces of $\tau$.  Because the
$\alpha_i$ have pairwise disjoint support, this can occur for at most $k+2$ indices
$i$.  Since all syndromes are supported in $F$,
\[
 \sum_{i=1}^q\norm{\delta_k\alpha_i}_{k+1}
 \le (k+2)\norm F_{k+1}.
\]
Combining the two inequalities proves \eqref{eq:packed}.

For sharpness, take a single $(k+1)$-simplex with its unique top face of weight one and
let the $\alpha_i$ be the singleton indicators of its $k+2$ facets.  Each has weight
$1/(k+2)$, all coboundaries equal the unique top face, and the expansion constant is
$k+2$.  Equality holds in \eqref{eq:packed}.
\end{proof}

\begin{remark}[Comparison with collective cosystolic expansion]
\label{rem:collective}
Collective cosystolic expansion, introduced by Kaufman and Tessler
\cite{kaufmantessler2021}, controls the size of a union of fillings in terms of a union of
coboundaries and is designed to behave well under tensor products.  Theorem~\ref{thm:packed}
asks a different adversarial question: how small can one common erased set be if it hides
many geometrically separated defects?  Neither statement formally implies the other
without additional assumptions.
\end{remark}

\section{Random erasures, enumerators, and matroid structure}
\label{sec:random}

We now pass from worst-case deletion to independent top-face erasure.  The unlocalized
rank law is matroidal, while localization-sensitive events require a finer bivariate
enumerator and the shattering hierarchy.

\subsection{The syndrome-localization enumerator}

\begin{definition}[Bivariate shattering enumerator]\label{def:enumerator}
Define
\begin{equation}\label{eq:enumerator}
 \cW_X^k(a,b)=\sum_{c\in\cC_X^{k+1}}
 a^{\lambda_X^k(c)}b^{\abs c}.
\end{equation}
\end{definition}

This polynomial refines the usual weight enumerator by recording the minimum size of a
preimage.  It is related to coset-leader and list-weight enumerators
\cite{jurriuspellikaan2015}, but its two coordinates live on adjacent dimensions of the
simplicial cochain complex.

\begin{proposition}[Basic properties of $\cW$]\label{prop:enumerator}
The enumerator satisfies:
\begin{enumerate}[label=\textup{(\roman*)},leftmargin=2.2em]
\item
$\Phi_X^k(s)=\min\set{t:[a^s b^t]\cW_X^k(a,b)\ne0}$ for $s\ge1$;
\item every nonconstant monomial $a^s b^t$ satisfies $t\ge h^k(X)s$;
\item if $X=X_1\sqcup X_2$ is a disjoint union, then
$\cW_X^k=\cW_{X_1}^k\cW_{X_2}^k$.
\end{enumerate}
\end{proposition}

\begin{proof}
Parts (i) and (ii) follow directly from the definitions and
\eqref{eq:h-profile}.  For (iii), the top coboundary code is a direct sum, support size is
additive, and the minimum preimage problem separates across the two components.
\end{proof}

\subsection{Independent top-face erasures}

The next identities separate two levels of randomness.  Rank increments depend only on
the matroid of $\delta_k$; localization-sensitive events depend on which shortened
subspaces occur and therefore require the bivariate enumerator.

Let $F_p\subseteq X(k+1)$ contain each top face independently with probability
$p\in[0,1]$; faces in $F_p$ are deleted.

\begin{theorem}[Exact Bernoulli-erasure identities]\label{thm:random-identities}
For every $p\in[0,1]$,
\begin{align}
 \E\bigl[\abs{\cE_{F_p}^k(X)}\bigr]
 &=\sum_{c\in\cC_X^{k+1}}p^{\abs c},
 \label{eq:expected-size}\\
 \E\left[\sum_{u\in\cE_{F_p}^k(X)}a^{\ell_X^k(u)}\right]
 &=\cW_X^k(a,p).
 \label{eq:expected-enumerator}
\end{align}
Here $\abs{\cE_{F_p}^k(X)}=2^{\Delta\beta_k(F_p)}$ denotes vector-space cardinality.
\end{theorem}

\begin{proof}
By Theorem~\ref{thm:exact-isomorphism}, an element of emergent cohomology corresponds to
a codeword $c$ whose support is contained in $F_p$.  This event has probability
$p^{\abs c}$.  Summing indicators over all codewords proves \eqref{eq:expected-size};
weighting each indicator by $a^{\lambda_X^k(c)}$ proves
\eqref{eq:expected-enumerator}.
\end{proof}

For a finite-dimensional $\F$-space of dimension $d$, let
$\qbinom{d}{q}$ denote the Gaussian binomial coefficient, the number of its
$q$-dimensional subspaces.  For a deletion set $F$, let
$N_{q,s}(F)$ be the number of subspaces $U\le\cE_F^k(X)$ with
$\dim U=q$ and $\ell_X^k(u)\ge s$ for every $0\ne u\in U$.

\begin{theorem}[Exact subspace moments and a shattering tail bound]
\label{thm:subspace-moments}
Let $r=\rank\delta_k$.  For $1\le q\le r$, $s\ge1$, and $p\in[0,1]$,
\begin{equation}\label{eq:localized-subspace-moment}
 \E\bigl[N_{q,s}(F_p)\bigr]
 =
 \sum_{\substack{U\le\cC_X^{k+1},\ \dim U=q\\
                   \lambda_X^k(c)\ge s\ \forall\,0\ne c\in U}}
 p^{\abs{\supp U}}.
\end{equation}
If $\Shat_{q,s}^k(X)<\infty$, then
\begin{equation}\label{eq:shattering-tail}
 \Prob\bigl(N_{q,s}(F_p)\ge1\bigr)
 \le
 \qbinom{r}{q}\,p^{\Shat_{q,s}^k(X)}.
\end{equation}
If $\Shat_{q,s}^k(X)=\infty$, the event on the left is impossible.  For $s=1$, the exact
identity becomes
\begin{equation}\label{eq:gaussian-moment}
 \E\left[\qbinom{\Delta\beta_k(F_p)}{q}\right]
 =
 \sum_{\substack{U\le\cC_X^{k+1}\\\dim U=q}}
 p^{\abs{\supp U}}.
\end{equation}
\end{theorem}

\begin{proof}
Under Theorem~\ref{thm:exact-isomorphism}, a fixed subcode
$U\le\cC_X^{k+1}$ occurs as a subspace of emergent cohomology exactly when
$\supp U\subseteq F_p$.  This event has probability $p^{\abs{\supp U}}$.
Summing its indicator over all localization-feasible $q$-dimensional subcodes proves
\eqref{eq:localized-subspace-moment}.

The event $N_{q,s}(F_p)\ge1$ has probability at most its expectation.  There are at most
$\qbinom{r}{q}$ candidate $q$-subspaces, and every feasible one has support at least
$\Shat_{q,s}^k(X)$, proving \eqref{eq:shattering-tail}.  When $s=1$, every
$q$-subspace of the shortened code is feasible, and a vector space of dimension
$\Delta\beta_k(F_p)$ has exactly $\qbinom{\Delta\beta_k(F_p)}{q}$ such
subspaces.  This gives \eqref{eq:gaussian-moment}.
\end{proof}

Combining \eqref{eq:shattering-tail} with Theorem~\ref{thm:griesmer-shat} turns either
the complete cofilling envelope or a scalar expansion constant into a random-erasure
tail estimate for the simultaneous creation of several delocalized classes.

\begin{theorem}[Macroscopic-defect probability bound]\label{thm:probability}
Let $N_k=\abs{X(k)}$.  For integers $1\le s\le t$ and $0<p\le1$, let
$I_{s,t}=\{j\in\{s,\ldots,t\}:\Phi_X^k(j)<\infty\}$.  Then
\begin{align}
 &\Prob\left(\exists\,0\ne u\in\cE_{F_p}^k(X):
 s\le\ell_X^k(u)\le t
 \right)\notag\\
 &\hspace{15mm}\le
 \sum_{j\in I_{s,t}}\binom{N_k}{j}p^{\Phi_X^k(j)}
 \le
 \sum_{j\in I_{s,t}}\binom{N_k}{j}p^{h^k(X)j}.
 \label{eq:probability}
\end{align}
\end{theorem}

\begin{proof}
For each codeword $c$, choose one minimum-weight preimage $\alpha(c)$ by a fixed
tie-breaking rule.  The map $c\mapsto\alpha(c)$ is injective because
$\delta_k\alpha(c)=c$.  Hence at most $\binom{N_k}{j}$ codewords have localization
$j$.  A codeword $c$ emerges exactly when its support is erased, an event of probability
$p^{\abs c}\le p^{\Phi_X^k(j)}$.  A union bound over $j$ proves the first inequality.
The second follows from $\Phi_X^k(j)\ge h^k(X)j$.
\end{proof}

Let $H_2(x)=-x\log_2x-(1-x)\log_2(1-x)$ be binary entropy.

\begin{corollary}[An entropy criterion]\label{cor:entropy}
Consider a sequence $X_n$ with $N_n=\abs{X_n(k)}\to\infty$ and
$h^k(X_n)\ge\varepsilon>0$.  Fix $0<\alpha\le\beta<1$ and $0<p<1$.  If
\begin{equation}\label{eq:entropy-condition}
 \sup_{x\in[\alpha,\beta]}
 \bigl(H_2(x)+\varepsilon x\log_2p\bigr)<0,
\end{equation}
then the probability that $F_p$ creates a nonzero quotient class with localization in
$[\alpha N_n,\beta N_n]$ is exponentially small in $N_n$.
\end{corollary}

\begin{proof}
Use $\binom{N_n}{j}\le2^{N_n H_2(j/N_n)}$ in
\eqref{eq:probability}.  Condition \eqref{eq:entropy-condition} makes every summand in
the indicated range exponentially small, uniformly in $j$; the number of summands is
only linear in $N_n$.
\end{proof}

\subsection{The full rank distribution as a Tutte evaluation}

By Corollary~\ref{cor:submodularity}, $\Delta\beta_k(F)$ is dual-matroid nullity.  The
next formula is therefore a standard rank-generating/Tutte specialization, in the same
matroid--code lineage as Greene's theorem and later support-matroid refinements
\cite{greene1976,barg1997,britz2007}.  We include it because it gives the exact
probability-generating function for topology created by independent erasures.

Let $m=\abs{X(k+1)}$, let $M=M_X^k$ be the row matroid of $\delta_k$, and write
$r=r_M(E)$.  We use the rank-generating definition
\[
 T_M(x,y)=\sum_{A\subseteq E}
 (x-1)^{r-r_M(A)}(y-1)^{\abs A-r_M(A)}.
\]

\begin{theorem}[Tutte specialization for the Betti increment]\label{thm:tutte}
For $0<p<1$,
\begin{equation}\label{eq:tutte}
 \E\left[z^{\Delta\beta_k(F_p)}\right]
 =p^{m-r}(1-p)^r
 T_M\left(1+\frac{pz}{1-p},\frac1p\right).
\end{equation}
The endpoint cases follow by continuity.
\end{theorem}

\begin{proof}
Put $A=E\setminus F_p$.  Corollary~\ref{cor:rank-defect} gives
$\Delta\beta_k(F_p)=r-r_M(A)$, while
\[
 \Prob(E\setminus F_p=A)=p^{m-\abs A}(1-p)^{\abs A}.
\]
Therefore
\[
 \E[z^{\Delta\beta_k(F_p)}]
 =\sum_{A\subseteq E}p^{m-\abs A}(1-p)^{\abs A}z^{r-r_M(A)}.
\]
Substituting
$x-1=pz/(1-p)$ and $y-1=(1-p)/p$ in the rank-generating formula and collecting the
prefactor gives \eqref{eq:tutte}.
\end{proof}

Threshold phenomena for random simplicial complexes are classical
\cite{linialmeshulam2006,meshulamwallach2009}, and Tutte-polynomial and
random-cluster descriptions of Bernoulli cell complexes are known in broader form
\cite{hiraokashirai2016}.  Theorem~\ref{thm:tutte} records the exact
specialization governing cohomology creation in the present top-face-erasure model.

\section{Examples, sharpness, and information loss}\label{sec:examples}

The examples below test the general bounds and identify which code summaries lose
map-coupled localization information.

\subsection{Standard check realizations and change of check basis}

\begin{theorem}[Pair-repetition code and extremal check-basis separation]
\label{thm:check-basis-separation}
Let
\[
 C_n=\set{(x,x):x\in\F^n}\subseteq\F^{2n},
 \qquad H_0=\begin{bmatrix}I_n&I_n\end{bmatrix}.
\]
Then $H_0$ is a full-rank parity-check matrix of the direct sum of $n$
binary repetition codes, and
\begin{equation}\label{eq:standard-pair-repetition}
 \Shat_{q,s}(H_0)=
 \begin{cases}
 \mathsf N_2(q,s),&\mathsf N_2(q,s)\le n,\\
 +\infty,&\mathsf N_2(q,s)>n.
 \end{cases}
\end{equation}
For every $q\ge1$ and $s\ge2$, put $n=\mathsf N_2(q,s)$.  There is
$T\in\operatorname{GL}(n,2)$ such that $H_1=TH_0$ is a row-equivalent
parity-check matrix of the same code and
\begin{equation}\label{eq:check-basis-gap}
 \Shat_{q,s}(H_1)=q,
 \qquad \Shat_{q,s}(H_0)=n.
\end{equation}
Thus the hierarchy is not invariant under change of check basis even when the
kernel code, domain and check lengths, rank, and image code are fixed.  For
$q\ge2$ the separation is genuinely collective.
\end{theorem}

\begin{proof}
For $y\in\F^n$, a preimage under $H_0$ is $(a,b)$ with $a+b=y$.
Coordinatewise $\abs a+\abs b\ge\abs y$, with equality at $(y,0)$, so
\begin{equation}\label{eq:H0-lambda}
 \lambda_{H_0}(y)=\abs y.
\end{equation}
A feasible $q$-space is therefore a binary code of distance at least $s$, and
its common support is its effective length.  This proves the lower bound in
\eqref{eq:standard-pair-repetition}; embedding a shortest
$[\mathsf N_2(q,s),q,d\ge s]$ code proves attainability.

Now set $n=\mathsf N_2(q,s)$ and choose a shortest binary $[n,q,d\ge s]$
code $V\le\F^n$.  Since $s\ge2$, $n>q$.  Choose
$T\in\operatorname{GL}(n,2)$ with
$T(V)=E_q:=\operatorname{span}\{e_1,\ldots,e_q\}$.  For $H_1=TH_0$,
\[
 \lambda_{H_1}(y)=\lambda_{H_0}(T^{-1}y)=\abs{T^{-1}y}.
\]
Hence every nonzero element of $E_q$ has localization at least $s$, so
$\Shat_{q,s}(H_1)\le q$.  A $q$-dimensional subspace needs at least $q$
support coordinates, proving equality.  Invertibility of $T$ gives
$\ker H_1=\ker H_0=C_n$.
\end{proof}

\begin{corollary}[Explicit collective separation at localization two]
\label{cor:explicit-q-separation}
For every $q\ge2$, put $n=q+1$ and
\[
 T_q=\begin{bmatrix}I_q&0\\ \mathbf1^{\mathsf T}&1\end{bmatrix}.
\]
Then $H_0=[I_n\ I_n]$ and $H_1=T_qH_0$ define the same pair-repetition code,
but
\[
 \Shat_{q,2}(H_0)=q+1,
 \qquad \Shat_{q,2}(H_1)=q.
\]
\end{corollary}

\begin{proof}
The Singleton bound gives $\mathsf N_2(q,2)\ge q+1$, and the even-weight
code of length $q+1$ attains equality.  Moreover,
$T_q^{-1}E_q=\set{(u,\1^{\mathsf T}u):u\in\F^q}$ is exactly that code.
Apply Theorem~\ref{thm:check-basis-separation}.
\end{proof}

\begin{remark}[Relation to stopping-set dependence]
Stopping distance also depends on the parity-check matrix
\cite{schwartzvardy2006}, but it concerns erased variable coordinates under
iterative decoding.  Here erasure acts on listed checks and feasibility
constrains a whole syndrome subspace.  Both settings show that row-equivalent
check presentations can have different operational robustness.
\end{remark}

\subsection{A single simplex and simplex boundaries}

These examples isolate the sharp coding mechanism.  The final coboundary matrix of a
simplex boundary is a complete-graph incidence matrix, so matching theory and the binary
simplex code make both localization and common support explicit.

\begin{example}[One top simplex]\label{ex:one-simplex}
Let $X$ consist of one $(k+1)$-simplex and all its faces.  Then the top coboundary code
has length and rank one.  Hence
\[
 \eta_1^k(X)=\Shat_{1,1}^k(X)=1,
 \qquad
 \Phi_X^k(1)=1.
\]
In normalized weights, the construction in the sharpness part of
Theorem~\ref{thm:packed} shows that all $k+2$ singleton $k$-face defects can be hidden by
the same top-face erasure.
\end{example}

\begin{theorem}[Boundary of a simplex]\label{thm:simplex-boundary}
Let $X=\partial\Delta^{k+2}$ and put $N=k+3$.  Then
\begin{align}
 \cC_X^{k+1}
 &=\set{x\in\F^N:\sum_{i=1}^N x_i=0},
 \label{eq:even-code}\\
 \Phi_X^k(s)&=2s
 &&\text{for }1\le s\le\floor{N/2},
 \label{eq:simplex-profile}\\
 \eta_q^k(X)&=q+1
 &&\text{for }1\le q\le N-1,
 \label{eq:simplex-eta}\\
 \cW_X^k(a,b)
 &=\sum_{s=0}^{\floor{N/2}}\binom{N}{2s}a^sb^{2s}.
 \label{eq:simplex-W}
\end{align}
Furthermore,
\begin{align}
 \Shat_{1,s}^k(X)&=2s
 &&\text{when }2s\le N,
 \label{eq:simplex-q1}\\
 \Shat_{2,s}^k(X)&=3s
 &&\text{when }3s\le N,
 \label{eq:simplex-q2}\\
 \Shat_{q,1}^k(X)&=q+1
 &&\text{when }1\le q\le N-1.
 \label{eq:simplex-sharp}
\end{align}
The first two quantities are $+\infty$ when the displayed feasibility inequalities
fail.  More generally, if $q\ge2$, $t\ge1$,
\begin{equation}\label{eq:simplex-parameter-family}
 s=t2^{q-2}
 \quad\text{and}\quad
 t(2^q-1)\le N,
\end{equation}
then
\begin{equation}\label{eq:simplex-general-sharp}
 \Shat_{q,s}^k(X)=t(2^q-1).
\end{equation}
Thus the profile-Griesmer bound is attained on infinite parameter families of simplex
boundaries.
\end{theorem}

\begin{proof}[Proof idea]
Label a top facet by its omitted vertex and a $k$-face by its two omitted vertices.  The
last coboundary matrix becomes the vertex--edge incidence matrix of $K_N$, so its image
is the even-weight code.  An even syndrome of size $2s$ has minimum preimage size $s$
by a matching argument.  The generalized weights are therefore $q+1$, while repeated
binary simplex-code columns attain the profile--Griesmer bound in
\eqref{eq:simplex-parameter-family}.  The complete verification, including the
infeasible cases, is in Appendix~\ref{app:simplex-proof}.
\end{proof}

\subsection{Complete graphs and Walsh--Fourier rigidity}

For a cut subspace of a complete graph, every vertex acquires a binary label.  The
objective is then determined by Walsh-character side sizes, separating the label
distribution problem from the geometry of the host graph and exposing the relevant
parity obstructions.

\begin{theorem}[Complete-graph shattering]\label{thm:complete-graph}
Let $1\le s\le\floor{n/2}$ and $q\ge1$.  Then
\begin{equation}\label{eq:complete-graph-griesmer}
 \Shat_{q,s}^0(K_n)\ge
 \mathsf G_q\bigl(s(n-s)\bigr).
\end{equation}
More precisely, every feasible $q$-dimensional cut subspace $U$ admits vertex labels
$x_v\in\F^q$ such that, writing
\[
 w_t=\abs{\{v:t\cdot x_v=1\}}
 \qquad(0\ne t\in\F^q),
\]
one has
\begin{equation}\label{eq:complete-graph-energy}
 \abs{\supp U}
 =2^{1-q}\sum_{0\ne t\in\F^q}w_t(n-w_t)
 \ge 2(1-2^{-q})s(n-s).
\end{equation}
If $2^{q-1}$ divides $s$, then equality holds in
\eqref{eq:complete-graph-griesmer}, and
\begin{equation}\label{eq:complete-graph-exact}
 \Shat_{q,s}^0(K_n)=2(1-2^{-q})s(n-s).
\end{equation}
\end{theorem}

\begin{proof}[Proof idea]
The graph profile is $\Phi_{K_n}^0(t)=t(n-t)$, so the universal lower bound immediately
gives \eqref{eq:complete-graph-griesmer}.  A basis of a feasible cut subspace labels the
vertices by $\F^q$; double-counting nonzero characters over unequal label pairs yields
\eqref{eq:complete-graph-energy}.  When $2^{q-1}\mid s$, repeating each nonzero label
equally makes every character side have size $s$, and the energy matches the Griesmer
bound.  Appendix~\ref{app:complete-graph-proof} gives the full construction and
divisibility check.
\end{proof}

\begin{remark}[Walsh multiplicities]\label{rem:walsh}
If $m_x=\abs{\{v:x_v=x\}}$ and
$\widehat m(t)=\sum_x m_x(-1)^{t\cdot x}$, then
$w_t=(n-\widehat m(t))/2$.  Walsh--Parseval gives the equivalent identity
\[
 \abs{\supp U}=\frac12\left(n^2-\sum_{x\in\F^q}m_x^2\right).
\]
Thus complete-graph shattering is an integer energy-minimization problem for a labelled
vertex partition subject to uniform bounds on all nontrivial Walsh coefficients.
\end{remark}

The maximal-localization endpoint has a complete solution in every dimension.  It is a
rigidity statement rather than only an extremal-energy calculation.

\begin{theorem}[Walsh rigidity at maximal localization]
\label{thm:complete-graph-maximal}
Let $n$ be even and $q\ge1$.  Then
\begin{equation}\label{eq:complete-maximal}
 \Shat_{q,n/2}^0(K_n)<\infty
 \quad\Longleftrightarrow\quad
 2^q\mid n.
\end{equation}
When these conditions hold, every feasible labeling uses each vector of $\F^q$ exactly
$n/2^q$ times, and
\begin{equation}\label{eq:complete-maximal-value}
 \Shat_{q,n/2}^0(K_n)
 =\frac{n^2}{2}\left(1-2^{-q}\right).
\end{equation}
\end{theorem}

\begin{proof}
For a feasible labeling, maximal localization forces
$w_t=n/2$ for every $0\ne t\in\F^q$.  In the notation of
Remark~\ref{rem:walsh}, this is equivalent to
$\widehat m(t)=0$ for all nonzero Walsh frequencies.  Fourier inversion on $\F^q$ then
gives
\[
 m_x=2^{-q}\widehat m(0)=\frac{n}{2^q}
 \qquad(x\in\F^q).
\]
Thus feasibility forces $2^q\mid n$ and a uniform label multiset.  Conversely, the
uniform labeling is feasible whenever $2^q\mid n$.  Substituting
$m_x=n/2^q$ into the Parseval energy identity in Remark~\ref{rem:walsh} gives
\[
 \abs{E_x}
 =\frac12\left(n^2-2^q\left(\frac{n}{2^q}\right)^2\right)
 =\frac{n^2}{2}\left(1-2^{-q}\right),
\]
which is therefore the unique feasible energy value and proves the theorem.
\end{proof}

The divisibility slices in Theorem~\ref{thm:complete-graph} are not the whole story.
For two-dimensional defect spaces the Walsh integer program can be solved for every
parameter, and parity produces a genuine feasibility obstruction.  The case $n=2s$ in
the next theorem is also the specialization $q=2$ of
Theorem~\ref{thm:complete-graph-maximal}.

\begin{theorem}[Exact two-dimensional shattering of $K_n$]
\label{thm:complete-graph-q2}
Let $1\le s\le\floor{n/2}$.  Then
\begin{equation}\label{eq:complete-q2}
 \Shat_{2,s}^0(K_n)=
 \begin{cases}
 +\infty,
   & n=2s\text{ and }s\text{ is odd},\\[1mm]
 \dfrac32s(n-s),
   & s\text{ is even, or }n\le3s\text{ and }n-s\text{ is even},\\[2mm]
 s(n-s)+\dfrac12(s+1)(n-s-1),
   & \text{otherwise}.
 \end{cases}
\end{equation}
\end{theorem}

\begin{proof}[Proof idea]
Let $a,b,c,d$ be the multiplicities of the four labels in $\F^2$, and let $u,v,w$ be
the one-set sizes of the three nonzero characters.  Fourier inversion gives the four
multiplicities from $(u,v,w)$, while the common cut support is
$\tfrac12\sum_{z\in\{u,v,w\}}z(n-z)$.  The concavity of $z(n-z)$ forces an optimum at
endpoint triples whenever parity permits; otherwise exactly one coordinate moves to the
nearest interior value.  The parity classification and the unique exceptional
infeasible case are worked out in Appendix~\ref{app:q2-proof}.
\end{proof}

\subsection{Cycles}

\begin{proposition}[Cycles]\label{prop:cycles}
For the cycle $C_n$, viewed in degree $k=0$,
\[
 \Phi_{C_n}^0(s)=2
 \quad (1\le s\le\floor{n/2}),
 \qquad
 \eta_q^0(C_n)=q+1
 \quad (1\le q\le n-1).
\]
Thus the rank-erasure hierarchy agrees with that of an even-weight code, while the
localization profile is flat rather than linear.
\end{proposition}

\begin{proof}
An interval of $s$ consecutive vertices has a two-edge boundary, and every nonempty
proper vertex set has boundary at least two.  This proves the profile formula by
Proposition~\ref{prop:graph-profile}.  The cut space of a cycle is the even-weight code
on its $n$ edges, giving the generalized weights $q+1$.
\end{proof}

The contrast between Theorem~\ref{thm:simplex-boundary} and
Proposition~\ref{prop:cycles} is the basic reason for retaining both parameters.  Their
rank-erasure weights are identical after matching code length, but their localization
profiles have completely different shapes.

The ordinary covering radius is not a matroid invariant
\cite{britzrutherford2005}.  The next result locates the analogous boundary for the
present hierarchy exactly: generalized-weight behavior survives at $s=1$, but every
positive localization demand beyond that boundary retains map-dependent information.

\begin{theorem}[Tutte boundary and localization-level non-invariance]
\label{thm:tutte-boundary-noninvariance}
Let $X$ be any finite $(k+1)$-complex.
\begin{enumerate}[label=\textup{(\roman*)},leftmargin=*]
\item For every $q$,
\[
 \Shat_{q,1}^k(X)=d_q(\cC_X^{k+1}).
\]
Hence the entire $s=1$ boundary is determined by the support matroid $M_X^k$ and is
recoverable from its Tutte polynomial.
\item For every integer $s\ge2$, there are connected graphs $G_s$ and $H_s$ of the
same order and size, together with an edge-coordinate bijection under which their binary
cut spaces are identical, but
\begin{equation}\label{eq:all-s-noninvariance}
 \Shat_{1,s}^0(G_s)=2,
 \qquad
 \Shat_{1,s}^0(H_s)=3.
\end{equation}
For $s=2$ the pair has five vertices and six edges; for every $s\ge3$ the pair can be
chosen with $2s$ vertices and $2s+1$ edges.  Consequently, for every fixed $s\ge2$, the level $\Shat_{1,s}^0$ is not determined by
the support matroid, its Tutte polynomial, the complete labeled image code, its
generalized-weight hierarchy, or its weight enumerator.
\end{enumerate}
\end{theorem}

\begin{proof}[Proof idea]
Part~(i) is the boundary identity in Theorem~\ref{thm:shat-code}.  Equivalently,
\[
 d_q(\cC_X^{k+1})=
 \min\set{\abs F:r(E)-r(E\setminus F)\ge q},
\]
so the rank-size distribution encoded by the Tutte polynomial determines every $d_q$
\cite{greene1976,britz2007}.

For part~(ii), the case $s=2$ is a five-vertex pair whose two cycle spaces coincide in
six labeled edge coordinates, while their minimum localization-two cuts have sizes two
and three.  For $s\ge3$, attach $s-2$ leaves to a common vertex in both graphs and
$s-3$ further leaves to different core vertices.  The new coordinates are bridges, so
the labeled cycle spaces, and hence the cut spaces, remain identical.  One graph has an
$s$-vertex cut of size two; a short case analysis of the core cuts of size at most two
shows that the other has minimum $s$-balanced cut size three.  The complete construction
and cut classification appear in Appendix~\ref{app:noninvariance-proof}.
\end{proof}

\begin{remark}[Why the second coordinate space is indispensable]
Part~(ii) keeps the domain length, image length, and the entire labeled image code
fixed.  What changes is the quotient metric induced by the incidence map from vertex
cochains to the common cut code.  Thus the map-coupled support/localization problem is
strictly finer than any invariant of the image code alone.
\end{remark}

\subsection{Face numbers and Betti numbers do not control shattering}

\begin{theorem}[Homological blindness already for graphs]
\label{thm:blindness}
Fix an integer $d\ge3$.  Along an infinite sequence of even $n$, there are connected
$d$-regular graphs $G_n$ and $H_n$ with identical face numbers and identical Betti
numbers as one-dimensional complexes, but
\begin{equation}\label{eq:blindness}
 \Phi_{G_n}^0(n/2)=\Omega(n),
 \qquad
 \Phi_{H_n}^0(n/2)\le2.
\end{equation}
Since $n/2$ is the largest possible localization in a connected graph,
$\Shat_{1,n/2}^0=\Sigma^0(n/2)=\Phi^0(n/2)$.  Hence no function of the face vector and
Betti numbers alone can determine, or provide a positive linear lower bound for,
middle-scale shattering.
\end{theorem}

\begin{proof}
Choose any infinite family of connected $d$-regular edge expanders $G_n$; for example,
random $d$-regular graphs have a uniform spectral and edge-expansion gap with high
probability along suitable sizes \cite{friedman2008}.  Then every bisection has
$\Omega(n)$ crossing edges, so Proposition~\ref{prop:graph-profile} gives the first
bound in \eqref{eq:blindness}.

For $H_n$, start with two connected $d$-regular graphs on $n/2$ vertices.  In each half,
choose a nonbridge edge, delete the two chosen edges, and cross-connect their four
endpoints with two edges.  The resulting graph is connected and $d$-regular, and the
partition into the two original halves is a bisection with two crossing edges.  This
gives the second bound.

Both graphs have $n$ vertices and $dn/2$ edges.  Since they are connected,
\[
 \beta_0=1,
 \qquad
 \beta_1=\frac{dn}{2}-n+1
\]
after viewing them as one-dimensional complexes.  Hence their face vectors and Betti
numbers agree.
\end{proof}

\section{Consequences, limitations, and open problems}\label{sec:discussion}

The map formulation separates three questions.  First, can a $q$-dimensional
syndrome space have coset-leader distance at least $s$?  The necessary
condition is
\[
 R_q(\ker A)\ge\mathsf N_2(q,s).
\]
Second, how many listed checks must support such a space?  Generalized weights
and the profile--Griesmer bound control that cost.  Third, is the answer
determined by the underlying code?  Theorem~\ref{thm:check-basis-separation}
shows that it is not: row-equivalent parity-check matrices of the same
nondegenerate code can realize the extremal values $q$ and
$\mathsf N_2(q,s)$.  The graph construction in
Theorem~\ref{thm:tutte-boundary-noninvariance} gives a structured incidence
version while keeping the complete labeled image code fixed.

This dependence is deliberate.  When check generators are physical cells,
local constraints, or specified parity equations, replacing them by arbitrary
linear combinations changes the failure model.  Proposition~\ref{prop:map-equivalence}
identifies the natural equivalence that preserves it: relabeling variables and
checks.  Choosing an optimal check basis for a fixed code is therefore a new
design problem suggested by the hierarchy.

The bounded-degree high-dimensional result remains qualitative.  It proves a
nonempty linear-order region but not the sharp range or density, and its
constants can be weak when upper incidence is large.  The weighted and CSS
interpretations likewise concern loss of listed check generators and relative
syndrome sectors, not physical-qubit erasure or absolute post-deletion quantum
distance.

\begin{openproblem}[Sharp macroscopic region]
For bounded-degree cosystolic-expander families, determine the closure of the
pairs $(\alpha,\beta)$ for which
$\Shat_{\floor{\alpha n},\ceil{\beta n}}^k(X)=\Theta(|X(k+1)|)$ and the
optimal asymptotic shattering density.
\end{openproblem}

\begin{openproblem}[Optimal check bases]
For a fixed binary code $C$ and prescribed number of check rows, determine
\[
 \min_H\Shat_{q,s}(H)\quad\text{and}\quad\max_H\Shat_{q,s}(H),
\]
where $H$ ranges over full-rank parity-check matrices of $C$.  Determine the
complexity of these problems and their relation to stopping redundancy and
generalized covering radii.
\end{openproblem}

\begin{openproblem}[Random check erasures]
For a sparse check realization $A_n$ and Bernoulli check-erasure parameter
$p$, determine the threshold for release of a $q_n$-dimensional syndrome
space whose every nonzero element has coset-leader weight at least $s_n$.
Determine when the syndrome-localization enumerator gives a sharp threshold.
\end{openproblem}

The main contribution is a realization-sensitive bridge between covering
radii and support hierarchies.  Generalized Hamming weights describe how
cheaply rank can be released; cofilling shattering asks how cheaply one can
release a subspace with no easy direction.  The pair-repetition theorem shows
that this robustness can change drastically with the check basis, while the
simplicial and graph results supply topological and cut-theoretic
realizations.

\appendix

\section{Proofs for the finite and asymptotic theory}
\label{app:finite-proofs}

This appendix contains the longer counting and limiting arguments used in the finite
and asymptotic bounds.

\subsection{Support--localization avoidance}
\label{app:localized-existence}

\begin{proof}[Full proof of Theorem~\ref{thm:localized-existence}]
Choose a $t$-dimensional subcode $W\le\cC_X^{k+1}$ with
$\abs{\supp W}=d_t(\cC_X^{k+1})$.  The map induced by $\delta_k$ is an isomorphism
\[
 \overline\delta_k:C^k(X)/Z^k(X)\longrightarrow\cC_X^{k+1}.
\]
Hence $V_W=\overline\delta_k^{-1}(W)$ has dimension $t$.  Endow it with the quotient
Hamming norm
\[
 \lambda(\alpha+Z^k(X))=\dist(\alpha,Z^k(X)).
\]
There are at most $\mathsf B(n,s)$ nonzero vectors of $V_W$ with norm below $s$:
choose one minimum-weight representative of each such coset.  The chosen representatives
are distinct cochains of weights $1,\ldots,s-1$.

Choose a $q$-dimensional subspace $U\le V_W$ uniformly at random.  For every fixed
nonzero $v\in V_W$, Gaussian-binomial counting gives
\[
 \Prob(v\in U)=
 \frac{\qbinom{t-1}{q-1}}{\qbinom tq}
 =\frac{2^q-1}{2^t-1}.
\]
The union bound and \eqref{eq:localized-existence-condition} show that some $U$ contains
no nonzero vector of norm below $s$.  Its image $\overline\delta_k(U)$ is feasible in
\eqref{eq:shat-code} and is supported inside $\supp W$, proving the first inequality in
\eqref{eq:localized-existence-conclusion}.

For the second inequality, choose $r$ coordinate positions that form an information set
for the $[m,r]$ code $\cC_X^{k+1}$.  Vanishing on any $r-t$ of those coordinates cuts
out a $t$-dimensional subcode supported on at most $m-r+t$ positions.  Therefore
$d_t(\cC_X^{k+1})\le m-r+t$.
\end{proof}

\subsection{Macroscopic bounded-degree shattering}
\label{app:macroscopic-shattering}

\begin{proof}[Full proof of Theorem~\ref{thm:macroscopic-shattering}]
Let $r_i=\rank\delta_k(X_i)$.  Fix $\tau$ with
\[
 H_2(\beta)+\alpha<\tau<\rho
\]
and set $t_i=\ceil{\tau n_i}$.  Lemma~\ref{lem:rank-density} gives
$r_i\ge\rho n_i$, so $q_i\le t_i\le r_i$ for all sufficiently large $i$.  Since
$s_i-1<\beta n_i$ and $\beta<1/2$, the standard entropy estimate gives
\[
 \mathsf B(n_i,s_i)\le2^{n_i H_2(\beta)}.
\]
Also $2^{t_i}-1\ge2^{t_i-1}$, and hence
\[
 \mathsf B(n_i,s_i)\frac{2^{q_i}-1}{2^{t_i}-1}
 \le2^{n_i(H_2(\beta)+\alpha-\tau)+1}<1
\]
for all sufficiently large $i$.  Theorem~\ref{thm:localized-existence} gives
\[
 \Shat_{q_i,s_i}^k(X_i)
 \le m_i-r_i+t_i
 \le m_i-(\rho-\tau)n_i+1.
\]
Purity implies $m_i\to\infty$, and Lemma~\ref{lem:rank-density} gives
$n_i/m_i\ge(k+2)/D$.  Divide by $m_i$, take the upper limit, and let
$\tau\downarrow H_2(\beta)+\alpha$ to obtain the upper bound.

Every $q_i$-dimensional binary code has support at least $q_i$.  Thus
Theorem~\ref{thm:griesmer-shat} gives
\[
 \Shat_{q_i,s_i}^k(X_i)\ge d_{q_i}(\cC_{X_i}^{k+1})\ge q_i.
\]
Dividing by $m_i$, using $q_i/n_i\to\alpha$, and applying
$n_i/m_i\ge(k+2)/D$ proves the first lower bound.  Under the expansion hypothesis,
Theorem~\ref{thm:griesmer-shat} also gives
\[
 \Shat_{q_i,s_i}^k(X_i)
 \ge\mathsf G_{q_i}(\ceil{\varepsilon s_i})
 \ge\ceil{\varepsilon s_i}.
\]
Combining the two lower estimates and using $s_i/n_i\to\beta$ proves
\eqref{eq:macroscopic-expanded-lower}.
\end{proof}

\section{Proof of the graph-labeling theorem}
\label{app:graph-proofs}

This appendix proves the labeling dictionary, Fourier energy identity, and spectral
bounds used in Section~\ref{sec:structural}.

\subsection{Fourier-balanced multiway cuts}
\label{app:fourier-proof}

\begin{proof}[Full proof of Theorem~\ref{thm:fourier-balanced-cut}]
Choose cochains $\alpha_1,\ldots,\alpha_q$ whose coboundaries form a basis of a cut
subspace $U$, and label each vertex by
\[
 x_v=(\alpha_1(v),\ldots,\alpha_q(v)).
\]
The linear map
\[
 t\longmapsto\delta_0\left(\sum_i t_i\alpha_i\right)
\]
is injective precisely when no nonzero affine functional is constant on all used labels,
equivalently when $\operatorname{affspan}x(V)=\F^q$.  For $0\ne t$, the associated
cut is $\partial_G S_t(x)$ and its localization is
$\min\{\abs{S_t(x)},n-\abs{S_t(x)}\}$.  The common support consists exactly of
edges whose endpoint labels differ.  This proves the minimization formula.

If $x_u\ne x_v$, exactly $2^{q-1}$ vectors $t\in\F^q$ satisfy
$t\cdot(x_u+x_v)=1$; the zero vector is not among them.  If $x_u=x_v$, no character
separates the endpoints.  Double-counting a cut together with one of its edges proves
\eqref{eq:graph-fourier-energy}.

For every $S\subseteq V$, the Rayleigh principle applied to
$\1_S-(\abs S/n)\1$ gives
\[
 \abs{\partial_G S}\ge
 \lambda_2(L_G)\frac{\abs S(n-\abs S)}{n}.
\]
Insert this bound into \eqref{eq:graph-fourier-energy} and use
$s\le\abs{S_t}\le n-s$ to obtain the first term of
\eqref{eq:graph-spectral-bounds}.

For the second term, let $V_1,\ldots,V_r$ be the nonempty label classes.  Full affine
span requires $r\ge q+1$.  The vectors
$z_j=\1_{V_j}/\sqrt{\abs{V_j}}$ are orthonormal, so the Ky Fan variational principle
gives
\[
 \sum_{j=1}^r z_j^{\mathsf T}L_G z_j
 \ge\sum_{i=1}^r\lambda_i(L_G).
\]
The left side equals
$\sum_j\abs{\partial_G V_j}/\abs{V_j}$, and each edge of $E_x$ contributes at most two
to this sum.  Therefore
$2\abs{E_x}\ge\sum_{i=2}^r\lambda_i(L_G)\ge\sum_{i=2}^{q+1}\lambda_i(L_G)$.
\end{proof}

\section{Deferred proofs for exact models and separation}
\label{app:exact-proofs}

This appendix gives the calculations for simplex boundaries, complete graphs, and the
graph families used in the non-invariance theorem.

\subsection{Simplex boundaries}
\label{app:simplex-proof}

\begin{proof}[Full proof of Theorem~\ref{thm:simplex-boundary}]
Label the $N$ top facets by the vertex omitted from $\Delta^{k+2}$ and label each
$k$-face by the pair of omitted vertices.  The matrix of $\delta_k$ is then the
vertex-edge incidence matrix of the complete graph $K_N$.  Its image is the even-weight
code, proving \eqref{eq:even-code}.

A syndrome $T\subseteq[N]$ of even cardinality $2s$ is the odd-degree set of an edge set
in $K_N$.  The minimum number of edges with odd-degree set $T$ is $s$: a matching on
$T$ attains $s$, while every edge accounts for at most two odd vertices.  Hence
$\lambda_X^k(\1_T)=s$, proving \eqref{eq:simplex-profile} and
\eqref{eq:simplex-W}.  The generalized weights of the length-$N$ even-weight code are
$q+1$, proving \eqref{eq:simplex-eta}.  In particular, $h^k(X)=2$.

Equation \eqref{eq:simplex-q1} follows from the first identity below, together with
\eqref{eq:simplex-profile}; equation \eqref{eq:simplex-sharp} follows from the second:
\[
 \Shat_{1,s}^k(X)=\Sigma_X^k(s),
 \qquad
 \Shat_{q,1}^k(X)=d_q(\cC_X^{k+1}).
\]

For the general construction, take the binary $q$-dimensional simplex code whose
generator matrix has each nonzero vector of $\F^q$ as one column, and repeat every
column $t$ times.  It has length $t(2^q-1)$, and every nonzero word has weight
$t2^{q-1}=2s$.  Since this weight is even, the code embeds into the length-$N$
even-weight code whenever \eqref{eq:simplex-parameter-family} holds.  This proves the
upper bound in \eqref{eq:simplex-general-sharp}.  The profile-Griesmer lower bound is
\[
 \mathsf G_q(2s)
 =\sum_{i=0}^{q-1}t2^{q-1-i}
 =t(2^q-1),
\]
so equality holds.  Taking $q=2$ gives \eqref{eq:simplex-q2}.  If $3s>N$, the same
Griesmer lower bound rules out a feasible two-dimensional subcode; if $2s>N$, the
profile contains no word with localization at least $s$.
\end{proof}

\subsection{Complete-graph shattering}
\label{app:complete-graph-proof}

\begin{proof}[Full proof of Theorem~\ref{thm:complete-graph}]
Proposition~\ref{prop:graph-profile} gives
$\Phi_{K_n}^0(t)=t(n-t)$ for $1\le t\le\floor{n/2}$.  This profile is increasing, so
$\Sigma_{K_n}^0(s)=s(n-s)$.  Theorem~\ref{thm:griesmer-shat} proves
\eqref{eq:complete-graph-griesmer}.

Choose a basis $\delta_0\alpha_1,\ldots,\delta_0\alpha_q$ of a feasible subspace
$U$.  Label each vertex by
\[
 x_v=(\alpha_1(v),\ldots,\alpha_q(v))\in\F^q.
\]
For $0\ne t\in\F^q$, the corresponding codeword is the cut of
$\{v:t\cdot x_v=1\}$, so it has $w_t(n-w_t)$ edges.  An edge $uv$ belongs to
$\supp U$ exactly when $x_u\ne x_v$.  In that case, precisely $2^{q-1}$ vectors $t$
satisfy $t\cdot(x_u+x_v)=1$.  Double-counting pairs consisting of a nonzero codeword
and an edge in its support gives \eqref{eq:complete-graph-energy}.  Feasibility implies
\[
 s\le w_t\le n-s,
 \qquad
 w_t(n-w_t)\ge s(n-s).
\]

Suppose $s=a2^{q-1}$.  Assign $a$ vertices each of the $2^q-1$ nonzero labels in
$\F^q$ and assign the remaining $n-a(2^q-1)$ vertices label zero.  Since $2s\le n$,
the zero label occurs at least $a$ times.  Every nonzero functional $t$ takes value one
on exactly $2^{q-1}$ nonzero labels, so every nonzero codeword has a side of size $s$.
The labels span $\F^q$, giving a $q$-dimensional subspace, and
\eqref{eq:complete-graph-energy} is an equality.  Finally, $s(n-s)$ is divisible by
$2^{q-1}$, so
\[
 \mathsf G_q\bigl(s(n-s)\bigr)
 =s(n-s)\sum_{i=0}^{q-1}2^{-i}
 =2(1-2^{-q})s(n-s).
\]
This matches the construction and proves \eqref{eq:complete-graph-exact}.
\end{proof}

\subsection{The exact two-dimensional formula}
\label{app:q2-proof}

\begin{proof}[Full proof of Theorem~\ref{thm:complete-graph-q2}]
Use labels $00,10,01,11$ with multiplicities $a,b,c,d$.  The three nonzero characters
have one-set sizes
\[
 u=b+d,\qquad v=c+d,\qquad w=b+c,
\]
and feasibility is exactly $s\le u,v,w\le n-s$.  Conversely,
\begin{equation}\label{eq:q2-inversion}
 \begin{aligned}
 2a&=2n-u-v-w,&\quad 2b&=u+w-v,\\
 2c&=v+w-u,& 2d&=u+v-w.
 \end{aligned}
\end{equation}
The energy identity gives
\begin{equation}\label{eq:q2-energy}
 \abs{\supp U}
 =\frac12\bigl(u(n-u)+v(n-v)+w(n-w)\bigr).
\end{equation}

The function $f(x)=x(n-x)$ is minimized on the integer interval $[s,n-s]$ at its
endpoints.  Hence the smallest possible value of \eqref{eq:q2-energy} is
$3s(n-s)/2$ if an endpoint triple $(u,v,w)\in\{s,n-s\}^3$ is realizable.  Up to
permutation, the possibilities are determined by the number $j$ of coordinates equal
to $n-s$.  Equation~\eqref{eq:q2-inversion} gives
\[
\begin{array}{c|c|c}
 j & \text{one resulting multiplicity pattern} & \text{feasibility condition}\\
\hline
0&(n-3s/2,s/2,s/2,s/2)&s\text{ even},\\
1&((n-s)/2,(n-s)/2,(3s-n)/2,(n-s)/2)
  &n\le3s,\ n-s\text{ even},\\
2&(s/2,s/2,s/2,(2n-3s)/2)&s\text{ even},\\
3&((3s-n)/2,(n-s)/2,(n-s)/2,(n-s)/2)
  &n\le3s,\ n-s\text{ even}.
\end{array}
\]
Here $n\ge2s$ makes every unlisted nonnegativity condition automatic.  Thus an endpoint
triple is realizable if and only if
\[
 s\text{ is even},
 \qquad\text{or}\qquad
 n\le3s\text{ and }n-s\text{ is even}.
\]

If $n=2s$ and $s$ is odd, the interval contains only the endpoint $s$, which is not
realizable, so no feasible two-space exists.  In every remaining case where no endpoint
triple is realizable, one has $s$ odd and $n>2s$.  At least one of $u,v,w$ is then
interior, so monotonicity of $f$ up to $n/2$ yields
\[
 \abs{\supp U}\ge\frac12\bigl(2f(s)+f(s+1)\bigr).
\]
Equality is attained by $(u,v,w)=(s,s,s+1)$: equation
\eqref{eq:q2-inversion} gives
\[
 (a,b,c,d)=
 \left(\frac{2n-3s-1}{2},\frac{s+1}{2},
       \frac{s+1}{2},\frac{s-1}{2}\right),
\]
which consists of nonnegative integers because $s$ is odd and $n>2s$.  Substitution in
\eqref{eq:q2-energy} gives the third line of \eqref{eq:complete-q2}.
\end{proof}

\subsection{Tutte boundary and localization-level non-invariance}
\label{app:noninvariance-proof}

\begin{proof}[Full proof of Theorem~\ref{thm:tutte-boundary-noninvariance}]
The identity $\Shat_{q,1}^k(X)=d_q(\cC_X^{k+1})$ is
\eqref{eq:shat-boundaries}.  Let $E=X(k+1)$ and let $r$ be the rank function of
$M_X^k$.  The shortening formula gives
\[
 d_q(\cC_X^{k+1})=
 \min\set{\abs F:F\subseteq E,\ r(E)-r(E\setminus F)\ge q}.
\]
The rank-generating expansion of the Tutte polynomial records the number of subsets
$A\subseteq E$ realizing each pair $(r(E)-r(A),\abs A-r(A))$ and hence each possible
value of $\abs A$.  It therefore determines
the displayed minima for every $q$.  This proves part~(i).

For part~(ii), Proposition~\ref{prop:graph-profile} and
\eqref{eq:shat-boundaries} identify $\Shat_{1,2}^0$ on the base pair with the
minimum cut having at least two vertices on each side.  For $s\ge3$, the enlarged
graphs have $2s$ vertices, so $\Shat_{1,s}^0$ is their minimum bisection size.  Begin
with the common vertex set
$V_0=\{a,b,c,d,\ell\}$ and identify six edge coordinates by
\begin{equation}\label{eq:base-graph-pair}
\begin{array}{c|cccccc}
 &e_1&e_2&e_3&e_4&e_5&e_6\\ \hline
G_2&ab&a\ell&ad&bc&bd&cd\\
H_2&ab&a\ell&bd&ac&ad&cd.
\end{array}
\end{equation}
Both graphs are connected, and in the displayed coordinates both cycle spaces equal
\[
 \operatorname{span}_{\F}
 \{e_1+e_3+e_5,\ e_4+e_5+e_6\}.
\]
The cut spaces, being the orthogonal complements of the cycle spaces, are therefore
identical.  In $G_2$, the set $\{a,\ell\}$ has boundary two, while every one-edge cut
isolates the single vertex $\ell$; hence $\Shat_{1,2}^0(G_2)=2$.  In $H_2$, the degree
sequence in the order $(a,b,c,d,\ell)$ is $(4,2,2,3,1)$.  For every two-vertex set
$\{u,v\}$,
\[
 \abs{\partial_{H_2}\{u,v\}}
 =\deg(u)+\deg(v)-2\mathbf 1_{uv\in E(H_2)}\ge3,
\]
and equality holds for $\{a,\ell\}$.  Thus
$\Shat_{1,2}^0(H_2)=3$.

Now fix $s\ge3$.  Form $G_s$ from $G_2$ by adjoining $s-2$ pendant leaves at $a$ and
$s-3$ pendant leaves at $b$.  Form $H_s$ from $H_2$ by adjoining $s-2$ pendant leaves
at $a$ and $s-3$ pendant leaves at $d$.  Match the new pendant-edge coordinates in any
bijection.  Both graphs have $2s$ vertices and $2s+1$ edges.  Every new coordinate is a
bridge and is therefore zero on the cycle space.  Since the base cycle spaces agree,
the enlarged labeled cycle spaces, and consequently the labeled cut spaces, agree.

In $G_s$, let $S_G$ consist of $a$, $\ell$, and all $s-2$ new leaves at $a$.  Then
$\abs{S_G}=s$ and only the base edges $ab$ and $ad$ leave $S_G$, so
$\abs{\partial_{G_s}S_G}=2$.  Every bridge of $G_s$ isolates a single pendant vertex
(including $a\ell$, which isolates $\ell$), so no one-edge cut has both sides of size at
least $s$.  Hence $\Shat_{1,s}^0(G_s)=2$.

The analogous set in $H_s$ has boundary $\{ab,ac,ad\}$, proving
$\Shat_{1,s}^0(H_s)\le3$.  It remains to exclude an $s$-vertex cut of size at most two.
Let $T=S\cap V_0$ for such a putative set $S$, and let $t$ be the number of pendant
edges crossing the cut.  Then
\[
 \abs{\partial_{H_s}S}=\abs{\partial_{H_2}T}+t.
\]
The core $H_2$ is connected.  Its cuts of size at most two, up to complementation, are
exactly
\[
 \varnothing\quad(0),
 \qquad \{\ell\}\quad(1),
 \qquad \{b\},\{c\}\quad(2),
\]
where parentheses give the core cut size.  Replacing $S$ by its complement if necessary,
assume that $T$ is one of the displayed representatives.  If the core cut is zero, then
$T=\varnothing$, and obtaining $\abs S=s$ requires at least $s$ pendant edges to cross,
contrary to $\abs{\partial_{H_s}S}\le2$.  If the core cut is one, then
$T=\{\ell\}$ and $t\le1$; since neither pendant cluster is attached to $\ell$, the set
$S$ contains at most $\ell$ and one pendant vertex, so $\abs S\le2<s$.  If the core cut
is two, then $T=\{b\}$ or $\{c\}$ and $t=0$; neither vertex carries pendant leaves in
$H_s$, so $S=T$ and $\abs S=1<s$.  Each case contradicts $\abs S=s$.  Therefore
$\Shat_{1,s}^0(H_s)=3$, completing part~(ii).
\end{proof}

\section*{Acknowledgements}
Generative-AI tools assisted with literature-search organization,
and preliminary checks.  The author reviewed the mathematical claims and citations and
accepts full responsibility for the final manuscript.

\section*{Statements and declarations}
\paragraph{Author contribution.}
Joshua Steier is the sole author.

\paragraph{Data availability.}
No external datasets were generated or analyzed.  The results are theoretical.

\end{document}